\documentclass[10pt,twocolumn,twoside]{IEEEtran}
\usepackage[cmex10]{amsmath}
\usepackage{amsthm, amsfonts, amssymb, graphicx, subfigure, texdraw, cite}

\theoremstyle{plain}
\newtheorem{proposition}{Proposition}
\newtheorem{lemma}{Lemma}
\newtheorem{theorem}{Theorem}
\newtheorem{corollary}{Corollary}
\theoremstyle{definition}
\newtheorem{definition}{Definition}

\newtheorem{algorithm}{Algorithm}
\theoremstyle{remark}
\newtheorem{remark}{Remark}
\newtheorem{example}{Example}

\newenvironment{algorithminit}[1]{\ \\{\em Initialization}: #1\begin{list}{\labelenumi}{\topsep0in\itemsep0in\parsep0in\labelwidth1in\usecounter{enumi}}}{\setcounter{enumii}{\value{enumi}}\end{list}}
\newenvironment{algorithmoper}[1]{{\em Operation}: #1\begin{list}{\labelenumi}{\topsep0in\itemsep0in\parsep0in\labelwidth1in\usecounter{enumi}\setcounter{enumi}{\value{enumii}}}}{\hfill$\blacksquare$\end{list}}

\begin{document}

\title{\huge A Distributed Algorithm for Solving Positive Definite Linear Equations over Networks with Membership Dynamics}
\author{Jie~Lu, \IEEEmembership{Member, IEEE}, and~Choon~Yik~Tang, \IEEEmembership{Member, IEEE}\thanks{This work was supported by the National Science Foundation under grant CMMI-0900806.}\thanks{J. Lu is with the School of Information Science and Technology, ShanghaiTech University, Shanghai 200031, China (e-mail: \mbox{lujie@shanghaitech.edu.cn}).}\thanks{C. Y. Tang is with the School of Electrical and Computer Engineering, University of Oklahoma, Norman, OK 73019 USA (e-mail: \mbox{cytang@ou.edu}).}}
\maketitle

\begin{abstract}
This paper considers the problem of solving a symmetric positive definite system of linear equations over a network of agents with arbitrary asynchronous interactions and membership dynamics. The latter implies that each agent is allowed to join and leave the network at any time, for infinitely many times, and lose all its memory upon leaving. We develop {\em Subset Equalizing} (SE), a distributed asynchronous algorithm for solving such a problem. To design and analyze SE, we introduce a novel time-varying Lyapunov-like function, defined on a state space with changing dimension, and a generalized concept of network connectivity, capable of handling such interactions and membership dynamics. Based on them, we establish the boundedness, asymptotic convergence, and exponential convergence of SE, along with a bound on its convergence rate. Finally, through extensive simulation, we show that SE is effective in a volatile agent network and that a special case of SE, termed {\em Groupwise Equalizing}, is significantly more bandwidth/\linebreak[0]energy efficient than two existing algorithms in multi-hop wireless networks.
\end{abstract}

\section{Introduction}\label{sec:intr}

\IEEEPARstart{S}{olving} a system of linear equations $Pz=q$ is a fundamental problem with numerous applications in science and engineering. In this paper, we address the problem of decentralizedly solving such equations over a network of $N$ agents, whereby each agent $i$ observes a symmetric positive definite matrix $P_i\in\mathbb{R}^{n\times n}$ and a vector $q_i\in\mathbb{R}^n$, and {\em all} of them wish to find the unique solution $z\in\mathbb{R}^n$ to
\begin{align}
\Bigl(\sum_{i=1}^NP_i\Bigr)z=\sum_{i=1}^Nq_i.\label{eq:sumPz=sumqintro}
\end{align}
The need to solve \eqref{eq:sumPz=sumqintro} arises in many applications of multi-agent systems. For instance, finding the maximum-likelihood estimate of an unknown parameter from noisy linear measurements in a wireless sensor network is equivalent to solving \eqref{eq:sumPz=sumqintro} over the network \cite{XiaoL05, XiaoL06}. Also, the widely studied average consensus problem (e.g., \cite{Tsitsiklis84, Olfati-Saber04, Boyd06, ChenJY06, Cortes06, Fagnani08, Olshevsky09, Oreshkin10, Tahbaz-Salehi10, LuJ12}) is a notable special case of \eqref{eq:sumPz=sumqintro} with $n=1$ and $P_i=1$ for all $i$.

Given its broad applications, problem \eqref{eq:sumPz=sumqintro} has received considerable attention in the literature. Most of the studies, however, focus on the special case of average consensus, as is evident by the rich collection of continuous-time (e.g., \cite{Olfati-Saber04, Cortes06}), discrete-time synchronous (e.g., \cite{Olfati-Saber04, Fagnani08, Olshevsky09, Oreshkin10, Tahbaz-Salehi10}), and discrete-time asynchronous (e.g., \cite{Tsitsiklis84, Boyd06, ChenJY06, Fagnani08, LuJ12}) algorithms that are available to date. Nonetheless, a few distributed algorithms devoted to the regular case of \eqref{eq:sumPz=sumqintro} with arbitrary $n$ and $P_i$'s have been proposed, including the continuous-time algorithm from \cite{Spanos05c}, which computes the solution $z$ to \eqref{eq:sumPz=sumqintro} by exploiting the positive definiteness of the $P_i$'s, and the two discrete-time synchronous, average-consensus-based algorithms from \cite{XiaoL05, XiaoL06}, which do so by element-wise averaging the $P_i$'s and $q_i$'s. Moreover, since problem \eqref{eq:sumPz=sumqintro} can be viewed as an unconstrained convex quadratic program, it may be solved using existing distributed convex optimization algorithms, including for example the quantized cyclic incremental \cite{Rabbat05}, subgradient-plus-consensus \cite{Nedic09}, primal-dual subgradient \cite{ZhuM12}, zero-gradient-sum \cite{LuJ12b}, adaptive penalty-based \cite{Towfic14}, and mixed-continuous/\linebreak[0]discrete-time \cite{Kia15} algorithms. There is also a related line of work that focuses on solving linear equations $Pz=q$ over an agent network, where each agent knows certain rows of $P$ and $q$ (e.g., \cite{MouS15}).

In this paper, we aim at solving the regular case of \eqref{eq:sumPz=sumqintro} over an agent network with dynamic memberships and topologies. Our development---from modeling to results---is substantially different from those in \cite{Spanos05c, XiaoL05, XiaoL06, Rabbat05, Nedic09, ZhuM12, LuJ12b, Towfic14, Kia15}, and generalizes some of those in \cite{Tsitsiklis84, Olfati-Saber04, Boyd06, ChenJY06, Cortes06, Fagnani08, Olshevsky09, Oreshkin10, Tahbaz-Salehi10, LuJ12} for average consensus. Specifically, we first introduce, in Section~\ref{sec:netwmodeprobform}, a novel agent network model that can handle arbitrary asynchronous interactions and membership dynamics, so that agents may freely interact with one another or spontaneously join and leave the network at any time, for infinitely many times. Unlike existing models in \cite{Tsitsiklis84, Olfati-Saber04, Spanos05c, XiaoL05, Boyd06, ChenJY06, Cortes06, XiaoL06, Fagnani08, Olshevsky09, Oreshkin10, Tahbaz-Salehi10, LuJ12, Rabbat05, Nedic09, ZhuM12, LuJ12b, Towfic14, Kia15} that require fixed agent memberships (i.e., graphs with fixed vertex sets), this model can handle dynamic ones, making it more general and allowing it to cope with practical situations, where agents may join or leave the network during runtime, either temporarily or permanently, voluntarily or involuntarily.

We next construct, in Section~\ref{sec:SE}, a distributed asynchronous algorithm named {\em Subset Equalizing} (SE) that enables the agents to cooperatively solve \eqref{eq:sumPz=sumqintro} despite having no control over their actions, essentially no knowledge about the network, and having to lose all their memories upon leaving the network. The algorithm SE is derived from a time-varying Lyapunov-like function that quantifies how far away the agents are from solving \eqref{eq:sumPz=sumqintro}, and from repeated minimization of this function in hope of incrementally dropping its value to zero. The algorithm is named SE because it is a networked dynamical system that evolves by repeatedly equalizing different subsets of its state variables. We also show that SE can be tailored to multi-hop networks with fixed vertex sets, leading to a gossip version called {\em Pairwise Equalizing} (PE) and a local broadcast version called {\em Groupwise Equalizing} (GE), which happen to generalize three existing average consensus schemes known as {\em Pairwise Averaging} \cite{Tsitsiklis84}, {\em Randomized Gossip Algorithm} \cite{Boyd06}, and {\em Distributed Random Grouping} \cite{ChenJY06}.

To analyze SE, we subsequently develop, in Section~\ref{sec:netwconn}, a few brand new notions of network connectivity---including instantaneous connectivity, connectivity, and uniform connectivity---which, unlike those in basic graph theory, are applicable to the agent network model (and which, together with the model, might be of interest in their own right). We also clarify these notions via examples and show that one of them generalizes a classic and widely used definition of connectivity for networks with fixed vertex sets and time-varying topologies, originally proposed in \cite{Tsitsiklis84}. Building upon these notions, we then derive, in Section~\ref{sec:bounconv}, sufficient conditions for establishing the boundedness, asymptotic convergence, and exponential convergence of SE, as well as a bound on its convergence rate. As a highlight of the results, we show that connectivity leads to asymptotic convergence, while uniform connectivity leads to exponential convergence.

As additional contributions of this paper, we demonstrate through simulation in Section~\ref{sec:simustud}, that SE is effective in a volatile agent network, while GE is several times more bandwidth/\linebreak[0]energy efficient than PE and the two algorithms from \cite{XiaoL05, XiaoL06} in multi-hop wireless networks. Finally, we state in Section~\ref{sec:conc} the conclusion of this paper. We note that this paper is an improved version of \cite{LuJ09, LuJ09b}. In addition, it contains in the Appendix all the proofs which are omitted in \cite{LuJ09, LuJ09b}. Throughout the paper, we let $\mathbb{N}$, $\mathbb{P}$, $\mathbb{S}_+^n$, and $|\cdot|$ denote, respectively, the sets of nonnegative integers, positive integers, $n\times n$ real symmetric positive definite matrices, and the cardinality of a set.

\section{Network Modeling and Problem Formulation}\label{sec:netwmodeprobform}

Consider a nonempty, finite set of $M\ge2$ agents, taking actions at each time $k\in\mathbb{N}$ according to the following model:
\begin{enumerate}
\renewcommand{\theenumi}{A\arabic{enumi}}\itemsep-\parsep
\item At time $k=0$, a nonempty subset $\mathcal{F}$ of the $M$ agents form a network and become {\em members} of the network.\label{enu:k=0}
\item Upon forming, each member $i\in\mathcal{F}$ observes a matrix $P_i\in\mathbb{S}_+^n$ and a vector $q_i\in\mathbb{R}^n$.\label{enu:obse}
\item The rest of the $M$ agents become {\em non-members} of the network and make no observations.\label{enu:k=0rest}
\item At each time $k\in\mathbb{P}$, three disjoint subsets of the $M$ agents---namely, a possibly empty subset $\mathcal{J}(k)$ of the non-members, a nonempty subset $\mathcal{I}(k)$ of the members, and a possibly empty, proper subset $\mathcal{L}(k)$ of the members---take actions~\ref{enu:join}--\ref{enu:leav} below.\label{enu:k=12}
\item The set $\mathcal{J}(k)$ of non-members join the network and become members.\label{enu:join}
\item Upon joining, the set $\mathcal{J}(k)\cup\mathcal{I}(k)\cup\mathcal{L}(k)$ of members interact, sharing information with one another and acknowledging their joining (i.e., $\mathcal{J}(k)$), staying (i.e., $\mathcal{I}(k)$), and leaving (i.e., $\mathcal{L}(k)$).\label{enu:inte}
\item Upon interacting, the set $\mathcal{L}(k)$ of members leave the network and become non-members.\label{enu:leav}
\item The rest of the $M$ agents (i.e., the complement of $\mathcal{J}(k)\cup\mathcal{I}(k)\cup\mathcal{L}(k)$) take no actions.\label{enu:k=12rest}
\end{enumerate}

Actions~\ref{enu:k=0}--\ref{enu:k=12rest} above define a general agent network model, where: (i) initially, an arbitrary subset of the agents form the network (i.e., \ref{enu:k=0}) and make {\em one-time} observations (\ref{enu:obse}), but the rest of them do not (\ref{enu:k=0rest}); (ii) at each subsequent time, arbitrary subsets of the agents (\ref{enu:k=12}) spontaneously join the network (\ref{enu:join}), interact with one another (\ref{enu:inte}), and leave the network (\ref{enu:leav}); and (iii) the agents take actions asynchronously (\ref{enu:k=12rest}). With this model, $M$ represents the maximum number of members the network may have, and each agent at any time is either a member or a non-member, but may change membership infinitely often. Labeling the $M$ agents as $1,2,\ldots,M$ and letting $\mathcal{M}(k)\subset\{1,2,\ldots,M\}$ denote the set of members upon completing the actions at each time $k\in\mathbb{N}$, the membership dynamics may be expressed as
\begin{align}\label{eq:M=McupJ-L}
\begin{split}
\mathcal{M}(0)&=\mathcal{F},\\
\mathcal{M}(k)&=(\mathcal{M}(k-1)\cup\mathcal{J}(k))-\mathcal{L}(k),\quad\forall k\in\mathbb{P},
\end{split}
\end{align}
where, since $\mathcal{F}\ne\emptyset$ and $\mathcal{L}(k)\subsetneq\mathcal{M}(k-1)$ $\forall k\in\mathbb{P}$, the network always has at least one member, i.e., $\mathcal{M}(k)\ne\emptyset$ $\forall k\in\mathbb{N}$. Moreover, since $\mathcal{J}(k)$ and $\mathcal{L}(k)$ may be empty for some $k\in\mathbb{P}$ but $\mathcal{I}(k)\ne\emptyset$ $\forall k\in\mathbb{P}$, while there may not always be membership changes, there are always member interactions, among the agents in $\mathcal{J}(k)\cup\mathcal{I}(k)\cup\mathcal{L}(k)$ $\forall k\in\mathbb{P}$. Since the membership dynamics and the member interactions are completely characterized by the sets $\mathcal{F}$, $\mathcal{J}(k)$, $\mathcal{I}(k)$, and $\mathcal{L}(k)$ $\forall k\in\mathbb{P}$, the network is driven by a sequence $\mathcal{A}$ of agent actions given by
\begin{align}
\mathcal{A}=(\mathcal{F},\mathcal{J}(1),\mathcal{I}(1),\mathcal{L}(1),\mathcal{J}(2),\mathcal{I}(2),\mathcal{L}(2),\ldots).\label{eq:A=FJILJIL}
\end{align}

\begin{remark}\label{rem:graph}
Although it is common to model networks using graphs, we use the sets $\mathcal{F}$, $\mathcal{J}(k)$, $\mathcal{I}(k)$, and $\mathcal{L}(k)$ $\forall k\in\mathbb{P}$ to model the above agent network because they enable convenient handling of the membership dynamics. We note that in the absence of membership changes (i.e., $\mathcal{J}(k)=\mathcal{L}(k)=\emptyset$ $\forall k\in\mathbb{P}$)---which is the de facto assumption in the literature---specifying $\mathcal{F}$ and $\mathcal{I}(k)$ $\forall k\in\mathbb{P}$ or $\mathcal{A}$ in \eqref{eq:A=FJILJIL} is the same as specifying an interaction graph.
\end{remark}

\begin{remark}\label{rem:leave}
Since $\mathcal{I}(k)\ne\emptyset$ $\forall k\in\mathbb{P}$, before leaving the network agents in $\mathcal{L}(k)$ always get to ``talk'' to someone who stays. This may be viewed as a limitation of the above agent network model because ``quiet'' departure of agents is not allowed.
\end{remark}

Given the agent network modeled by~\ref{enu:k=0}--\ref{enu:k=12rest}, the objective of this paper is to design and analyze a distributed asynchronous algorithm of iterative nature, which allows the ever-changing members of the network to cooperatively and asymptotically compute the constant solution $z\in\mathbb{R}^n$ of the following symmetric positive definite system of linear equations, defined by the one-time observations $P_i$ and $q_i$ $\forall i\in\mathcal{M}(0)$ of the initial members:
\begin{align}
\Bigl(\sum_{i\in\mathcal{M}(0)}P_i\Bigr)z=\sum_{i\in\mathcal{M}(0)}q_i.\label{eq:sumPz=sumq}
\end{align}
The algorithm should also exhibit the following desirable properties:
\begin{enumerate}
\renewcommand{\theenumi}{P\arabic{enumi}}\itemsep-\parsep
\item It should allow the sequence $\mathcal{A}$ of agent actions to be dictated by an exogenous source, for which the agents have no control over, since, for example, in a sensor network, $\mathcal{J}(k)$, $\mathcal{I}(k)$, and $\mathcal{L}(k)$ may be governed by sensor redeployment, reseeding, mobility, failures, and recoveries, all of which may be forced exogenously.\label{enu:exog}
\item It should allow the agents to not know the values of $M$, $k$, $\mathcal{F}$, $\mathcal{J}(k)$, $\mathcal{I}(k)$, $\mathcal{L}(k)$, and $\mathcal{M}(k)$ $\forall k\in\mathbb{P}$, since in many practical situations they are not available, or at least not known ahead of time.\label{enu:notknow}
\item It should not impose large memory requirements on the agents, and should allow them to lose all their memories upon leaving the network, since the departure may be caused by, for instance, software or hardware failures.\label{enu:memo}
\end{enumerate}

\section{Subset Equalizing}\label{sec:SE}

In this section, using ideas from Lyapunov stability theory and optimization, we construct an algorithm that possesses properties~\ref{enu:exog}--\ref{enu:memo} and strives to solve \eqref{eq:sumPz=sumq}.

Consider a networked dynamical system formed by the $M$ agents, in which each agent $i\in\{1,2\ldots,M\}$ maintains in its memory two state variables $z_i\in\mathbb{R}^n\cup\{\#\}$ and $Q_i\in\mathbb{S}_+^n\cup\{\#\}$, where $z_i$ represents its estimate of the unknown solution $z$ of \eqref{eq:sumPz=sumq}, $Q_i$ plays the part of helping $z_i$ approach $z$, and the symbol $\#$ means {\em undefined}. To describe the system dynamics, let $z_i(k)$ and $Q_i(k)$ be the values of $z_i$ and $Q_i$ upon completing the actions at each time $k\in\mathbb{N}$. Let $z_i(k)\in\mathbb{R}^n$ and $Q_i(k)\in\mathbb{S}_+^n$ if $i\in\mathcal{M}(k)$, and $z_i(k)=\#$ and $Q_i(k)=\#$ otherwise.

Next, we specify the evolution of the state variables $z_i(k)$ and $Q_i(k)$ $\forall i\in\mathcal{M}(k)$ $\forall k\in\mathbb{N}$. To this end, consider a time-varying Lyapunov-like function $V$ of the $z_i(k)$'s and $Q_i(k)$'s, defined for each $k\in\mathbb{N}$ as
\begin{align}
&V(k,z_1(k),z_2(k),\ldots,z_M(k),Q_1(k),Q_2(k),\ldots,Q_M(k))\nonumber\displaybreak[0]\\
&\quad=\sum_{i\in\mathcal{M}(k)}(z_i(k)-z)^TQ_i(k)(z_i(k)-z).\label{eq:V}
\end{align}
Note that, as the left-hand side of \eqref{eq:V} is lengthy, we write it as $V(k)$ in the sequel for brevity. Also, as the right-hand side of \eqref{eq:V} excludes all the non-members $i\in\{1,2,\ldots,M\}-\mathcal{M}(k)$, $V(k)\in\mathbb{R}$ is always well-defined. Furthermore, as the sum involves a time-varying subset of the $z_i(k)$'s and $Q_i(k)$'s, $V(k)$ is akin to a function defined on a state space with growing and shrinking dimension. Finally, although not a standard Lyapunov function candidate, $V(k)$ exhibits some similar features that make it useful for the problem at hand: $V(k)\ge0$ $\forall k\in\mathbb{N}$, with $V(k)=0$ if and only if $z_i(k)=z$ $\forall i\in\mathcal{M}(k)$, i.e., \eqref{eq:sumPz=sumq} is exactly solved. This explains why we define it as such and call it a Lyapunov-{\em like} function.

Having introduced $V(k)$, we now use it to devise the system dynamics. To begin, observe from~\ref{enu:k=12}--\ref{enu:k=12rest} and \eqref{eq:M=McupJ-L} that for each $k\in\mathbb{P}$, the members in $\mathcal{M}(k)$ can be partitioned into those in $\mathcal{M}(k)-(\mathcal{J}(k)\cup\mathcal{I}(k))$ who take no actions at time $k$, and those in $\mathcal{J}(k)\cup\mathcal{I}(k)$ who interact with the leaving members in $\mathcal{L}(k)$. For those in $\mathcal{M}(k)-(\mathcal{J}(k)\cup\mathcal{I}(k))$, as they gain no new information, their $z_i$'s and $Q_i$'s are unchanged, i.e.,
\begin{align}
z_i(k)&=z_i(k-1),\quad\forall i\in\mathcal{M}(k)-(\mathcal{J}(k)\cup\mathcal{I}(k)),\label{eq:z=z}\displaybreak[0]\\
Q_i(k)&=Q_i(k-1),\quad\forall i\in\mathcal{M}(k)-(\mathcal{J}(k)\cup\mathcal{I}(k)).\label{eq:Q=Q}
\end{align}
For those in $\mathcal{J}(k)\cup\mathcal{I}(k)$, they get to jointly determine $z_i(k)$ and $Q_i(k)$ $\forall i\in\mathcal{J}(k)\cup\mathcal{I}(k)$ based on $z_i(k-1)$ and $Q_i(k-1)$ $\forall i\in\mathcal{I}(k)\cup\mathcal{L}(k)$. To enable such determination, notice from \eqref{eq:V}, \eqref{eq:M=McupJ-L}, \eqref{eq:z=z}, and \eqref{eq:Q=Q} that the change in the value of $V$ is
\begin{align}
&V(k)-V(k-1)=\Bigl[\sum_{i\in\mathcal{J}(k)\cup\mathcal{I}(k)}z_i(k)^TQ_i(k)z_i(k)\Bigr]\nonumber\displaybreak[0]\\
&\quad-\Bigl[\sum_{i\in\mathcal{I}(k)\cup\mathcal{L}(k)}z_i(k-1)^TQ_i(k-1)z_i(k-1)\Bigr]\nonumber\displaybreak[0]\\
&\quad-2z^T\Bigl[\!\!\!\!\!\sum_{i\in\mathcal{J}(k)\cup\mathcal{I}(k)}\!\!\!\!\!Q_i(k)z_i(k)-\!\!\!\!\!\sum_{i\in\mathcal{I}(k)\cup\mathcal{L}(k)}\!\!\!\!\!Q_i(k-1)z_i(k-1)\Bigr]\nonumber\displaybreak[0]\\
&\quad+z^T\Bigl[\sum_{i\in\mathcal{J}(k)\cup\mathcal{I}(k)}Q_i(k)-\sum_{i\in\mathcal{I}(k)\cup\mathcal{L}(k)}Q_i(k-1)\Bigr]z.\label{eq:V-V}
\end{align}
Also note that $V(k)-V(k-1)$ in \eqref{eq:V-V} would be unaffected by the unknown $z$ and thus would be known to the members in $\mathcal{J}(k)\cup\mathcal{I}(k)$ if the to-be-determined variables are chosen such that the third and fourth brackets in \eqref{eq:V-V} disappear, i.e.,
\begin{align}
\sum_{i\in\mathcal{I}(k)\cup\mathcal{L}(k)}\!\!\!\!Q_i(k-1)z_i(k-1)&=\sum_{i\in\mathcal{J}(k)\cup\mathcal{I}(k)}\!\!\!\!Q_i(k)z_i(k),\label{eq:sumQz=sumQz}\displaybreak[0]\\
\sum_{i\in\mathcal{I}(k)\cup\mathcal{L}(k)}Q_i(k-1)&=\sum_{i\in\mathcal{J}(k)\cup\mathcal{I}(k)}Q_i(k).\label{eq:sumQ=sumQ}
\end{align}
Moreover, as the second bracket is fixed, $V(k)-V(k-1)$ in \eqref{eq:V-V} would be minimized, perhaps even made negative, by having those members jointly minimize the first bracket, i.e.,
\begin{align}
\begin{array}{cl}\displaystyle\operatornamewithlimits{minimize}_{(z_i(k),Q_i(k))_{i\in\mathcal{J}(k)\cup\mathcal{I}(k)}} & \displaystyle\sum_{i\in\mathcal{J}(k)\cup\mathcal{I}(k)}\!\!\!\!z_i(k)^TQ_i(k)z_i(k)\\ \operatorname{subject\,to} & \eqref{eq:sumQz=sumQz}\;\text{and}\;\eqref{eq:sumQ=sumQ}.\end{array}\label{eq:minsumzQz}
\end{align}

\begin{lemma}\label{lem:optnonincr}
For any $\mathcal{A}$ and $k\in\mathbb{P}$, $(z_i(k),Q_i(k))_{i\in\mathcal{J}(k)\cup\mathcal{I}(k)}$ is an optimal solution to problem \eqref{eq:minsumzQz} if and only if $Q_i(k)$ $\forall i\in\mathcal{J}(k)\cup\mathcal{I}(k)$ satisfy \eqref{eq:sumQ=sumQ} and
\begin{align}
z_i(k)=\Bigl(\!\!\!\!\!\sum_{j\in\mathcal{I}(k)\cup\mathcal{L}(k)}\!\!\!\!\!Q_j(k-1)\Bigr)^{-1}&\!\!\!\!\!\sum_{j\in\mathcal{I}(k)\cup\mathcal{L}(k)}\!\!\!\!\!Q_j(k-1)z_j(k-1),\nonumber\displaybreak[0]\\
&\forall i\in\mathcal{J}(k)\cup\mathcal{I}(k).\label{eq:zopti}
\end{align}
Moreover, if \eqref{eq:z=z}, \eqref{eq:Q=Q}, \eqref{eq:sumQ=sumQ}, and \eqref{eq:zopti} hold, then $V(k)\le V(k-1)$, where the equality holds if and only if $z_i(k-1)$ $\forall i\in\mathcal{I}(k)\cup\mathcal{L}(k)$ are equal.
\end{lemma}

Lemma~\ref{lem:optnonincr} says that the optimal solution to \eqref{eq:minsumzQz} is an {\em equalizing} action, whereby the $z_i(k)$'s of the members in $\mathcal{J}(k)\cup\mathcal{I}(k)$ are set equal to the {\em same} value given by \eqref{eq:zopti}. Indeed, this equalizing action \eqref{eq:zopti}, along with \eqref{eq:sumQ=sumQ}, enables the agents in $\mathcal{J}(k)\cup\mathcal{I}(k)\cup\mathcal{L}(k)$ to jointly make the value of $V$ decrease, unless the $z_i(k-1)$'s of those in $\mathcal{I}(k)\cup\mathcal{L}(k)$ are identical, in which case the value of $V$ is unchanged.

Although $z_i(k)$ $\forall i\in\mathcal{J}(k)\cup\mathcal{I}(k)$ are uniquely determined by \eqref{eq:zopti}, there are infinitely many ways for $Q_i(k)$ $\forall i\in\mathcal{J}(k)\cup\mathcal{I}(k)$ to satisfy \eqref{eq:sumQ=sumQ}. For simplicity, we adopt the following way to determine $Q_i(k)$ $\forall i\in\mathcal{J}(k)\cup\mathcal{I}(k)$ so that \eqref{eq:sumQ=sumQ} holds: when there are no membership changes, i.e., $\mathcal{J}(k)=\mathcal{L}(k)=\emptyset$, the members in $\mathcal{J}(k)\cup\mathcal{I}(k)$ do not update their $Q_i(k)$'s, i.e., $Q_i(k)=Q_i(k-1)$ $\forall i\in\mathcal{J}(k)\cup\mathcal{I}(k)$, whereas when there are membership changes, i.e., $\mathcal{J}(k)\cup\mathcal{L}(k)\ne\emptyset$, their $Q_i(k)$'s are set equal to the {\em same} value while satisfying \eqref{eq:sumQ=sumQ}, i.e., $Q_i(k)=\frac{1}{|\mathcal{J}(k)\cup\mathcal{I}(k)|}\sum_{j\in\mathcal{I}(k)\cup\mathcal{L}(k)}Q_j(k-1)$ $\forall i\in\mathcal{J}(k)\cup\mathcal{I}(k)$.

Having specified the evolution of the state variables $z_i(k)$'s and $Q_i(k)$'s, we next define the initial states $z_i(0)$ and $Q_i(0)$ $\forall i\in\mathcal{M}(0)$. Notice from \eqref{eq:z=z}, \eqref{eq:Q=Q}, \eqref{eq:sumQz=sumQz}, and \eqref{eq:sumQ=sumQ} that
\begin{align}
\sum_{i\in\mathcal{M}(k)}Q_i(k)z_i(k)&=\sum_{i\in\mathcal{M}(0)}Q_i(0)z_i(0),\quad\forall k\in\mathbb{N},\label{eq:sumQz=sumQz0}\displaybreak[0]\\
\sum_{i\in\mathcal{M}(k)}Q_i(k)&=\sum_{i\in\mathcal{M}(0)}Q_i(0),\quad\forall k\in\mathbb{N}.\label{eq:sumQ=sumQ0}
\end{align}
Also note that problem \eqref{eq:sumPz=sumq} is solved only if $z_i(k)$ $\forall i\in\mathcal{M}(k)$ asymptotically reach a consensus. Hence, the consensus is the solution $z$ of problem~\eqref{eq:sumPz=sumq} if $\sum_{i\in\mathcal{M}(0)}Q_i(0)z_i(0)=\sum_{i\in\mathcal{M}(0)}q_i$ and $\sum_{i\in\mathcal{M}(0)}Q_i(0)=\sum_{i\in\mathcal{M}(0)}P_i$. To satisfy these two equations, it suffices to let $z_i(0)=P_i^{-1}q_i$ and $Q_i(0)=P_i$, which can be locally realized by each initial member $i\in\mathcal{M}(0)$.

The above expressions define a distributed asynchronous iterative algorithm. Since at each time $k\in\mathbb{P}$, this algorithm involves an {\em equalizing} action taken by a {\em subset} $\mathcal{J}(k)\cup\mathcal{I}(k)\cup\mathcal{L}(k)$ of the agents, we refer to the algorithm as {\em Subset Equalizing} (SE). A complete description of SE is as follows:

\begin{algorithm}[Subset Equalizing]\label{alg:SE}
\begin{algorithminit}{At time $k=0$:}
\item Each agent $i\in\{1,2,\ldots,M\}$ creates variables $z_i\in\mathbb{R}^n\cup\{\#\}$ and $Q_i\in\mathbb{S}_+^n\cup\{\#\}$ and initializes them as
\begin{align}
z_i(0)&=\begin{cases}P_i^{-1}q_i, & \text{if $i\in\mathcal{M}(0)$},\\ \#, & \text{otherwise},\end{cases}\label{eq:zinitial}\displaybreak[0]\\
Q_i(0)&=\begin{cases}P_i, & \text{if $i\in\mathcal{M}(0)$},\\ \#, & \text{otherwise}.\end{cases}\label{eq:Qinitial}
\end{align}
\end{algorithminit}
\begin{algorithmoper}{At each time $k\in\mathbb{P}$:}
\item Agents in $\mathcal{J}(k)$ join the network.
\item Each agent $i\in\{1,2,\ldots,M\}$ updates $z_i(k)$ according to
\begin{align}
z_i(k)\!=\!\begin{cases}\displaystyle\Bigl(\sum_{j\in\mathcal{I}(k)\cup\mathcal{L}(k)\!\!\!\!\!\!\!\!\!\!\!\!\!\!\!\!\!\!\!\!}Q_j(k-1)\Bigr)^{-1}\times\cdots & \\ \displaystyle\sum_{j\in\mathcal{I}(k)\cup\mathcal{L}(k)\!\!\!\!\!\!\!\!\!\!\!\!\!\!\!\!\!\!\!\!}Q_j(k\!-\!1)z_j(k\!-\!1), &\!\!\text{if $i\!\in\!\mathcal{J}(k)\!\cup\!\mathcal{I}(k)$},\\ \#, &\!\!\text{if $i\in\mathcal{L}(k)$},\\ z_i(k-1), &\!\!\text{otherwise}.\end{cases}\label{eq:z}
\end{align}
\item If $\mathcal{J}(k)=\mathcal{L}(k)=\emptyset$, then each agent $i\in\{1,2,\ldots,M\}$ updates $Q_i(k)$ according to
\begin{align}
Q_i(k)=Q_i(k-1),\quad\forall i\in\{1,2,\ldots,M\}.\label{eq:Qunchange}
\end{align}
Otherwise, each agent $i\in\{1,2,\ldots,M\}$ updates $Q_i(k)$ according to
\begin{align}
Q_i(k)\!=\!\begin{cases}\!\frac{1}{|\mathcal{J}(k)\cup\mathcal{I}(k)|}\displaystyle\!\sum_{j\in\mathcal{I}(k)\cup\mathcal{L}(k)\!\!\!\!\!\!\!\!\!\!\!\!\!\!\!\!\!\!\!\!}Q_j(k\!-\!1), &\!\!\text{if $i\!\in\!\mathcal{J}(k)\!\cup\!\mathcal{I}(k)$},\\ \#, &\!\!\text{if $i\in\mathcal{L}(k)$},\\ Q_i(k-1), &\!\!\text{otherwise}.\end{cases}\label{eq:Qchange}
\end{align}
\item Agents in $\mathcal{L}(k)$ leave the network.
\end{algorithmoper}
\end{algorithm}

Having presented SE, we next describe two ways that SE can be tailored to multi-hop networks with fixed vertex sets. Consider a multi-hop network modeled as an undirected, connected graph $\mathcal{G}=(\mathcal{V},\mathcal{E})$, where $\mathcal{V}=\{1,2,\ldots,N\}$ is the set of nodes, and $\mathcal{E}\subset\{\{i,j\}:i,j\in\mathcal{V},i\ne j\}$ is the set of edges. Suppose each node $i\in\mathcal{V}$ observes $P_i\in\mathbb{S}_+^n$ and $q_i\in\mathbb{R}^n$, and all of them wish to solve \eqref{eq:sumPz=sumqintro} for $z\in\mathbb{R}^n$. Also suppose the nodes wish to do so in one of the following two ways: (i) by having every node $i\in\mathcal{V}$ {\em gossip} with a neighbor $j\in\mathcal{N}_i=\{j\in\mathcal{V}:\{i,j\}\in\mathcal{E}\}$ from time to time, or (ii) by having every node $i\in\mathcal{V}$ {\em interact} with all its neighbors in $\mathcal{N}_i$ {\em as a group} from time to time. Note that this setup is a special case of the agent network~\ref{enu:k=0}--\ref{enu:k=12rest}, obtained by letting $M=N$, $\mathcal{F}=\mathcal{V}$, $\mathcal{J}(k)\equiv\emptyset$, $\mathcal{L}(k)\equiv\emptyset$, $\mathcal{I}(k)\in\mathcal{E}$ for~(i), and $\mathcal{I}(k)\in\{\{i\}\cup\mathcal{N}_i:i\in\mathcal{V}\}$ for~(ii). Thus, the nodes can solve \eqref{eq:sumPz=sumqintro} by using SE, which in this special case may be referred to as {\em Pairwise Equalizing} (PE) for~(i), and as {\em Groupwise Equalizing} (GE) for~(ii). A complete description of PE and GE that includes their communication and computation aspects is as follows (see \cite{LuJ09b} for more details):

\begin{algorithm}[Pairwise Equalizing]\label{alg:PE}
\begin{algorithminit}{}
\item Each node $i\in\mathcal{V}$ transmits $P_i$ to every node $j\in\mathcal{N}_i$.
\item Each node $i\in\mathcal{V}$ creates a variable $z_i\in\mathbb{R}^n$ and initializes it: $z_i\leftarrow P_i^{-1}q_i.$
\end{algorithminit}
\begin{algorithmoper}{At each iteration:}
\item A node, say, node $i\in\mathcal{V}$, initiates the iteration and selects a neighbor, say, node $j\in\mathcal{N}_i$, to gossip.
\item Node $i$ transmits $z_i$ to node $j$.
\item Node $j$ updates $z_j$: $z_j\leftarrow(P_i+P_j)^{-1}(P_iz_i+P_jz_j).$
\item Node $j$ transmits $z_j$ to node $i$.
\item Node $i$ updates $z_i$: $z_i\leftarrow z_j.$
\end{algorithmoper}
\end{algorithm}

\begin{algorithm}[Groupwise Equalizing]\label{alg:GE}
\begin{algorithminit}{}
\item Each node $i\in\mathcal{V}$ transmits $P_i$ to every node $j\in\mathcal{N}_i$.
\item Each node $i\in\mathcal{V}$ creates a variable $z_i\in\mathbb{R}^n$ and initializes it: $z_i\leftarrow P_i^{-1}q_i.$
\end{algorithminit}
\begin{algorithmoper}{At each iteration:}
\item A node, say, node $i\in\mathcal{V}$, initiates the iteration and transmits a message to every node $j\in\mathcal{N}_i$, requesting their $z_j$'s.
\item Each node $j\in\mathcal{N}_i$ transmits $z_j$ to node $i$.
\item Node $i$ updates $z_i$: $\displaystyle z_i\leftarrow\Bigl(\sum_{j\in\{i\}\cup\mathcal{N}_i\!\!\!\!\!}P_j\Bigr)^{-1}\!\!\sum_{j\in\{i\}\cup\mathcal{N}_i\!\!\!\!\!}P_jz_j.$
\item Node $i$ transmits $z_i$ to every node $j\in\mathcal{N}_i$.
\item Each node $j\in\mathcal{N}_i$ updates $z_j$: $z_j\leftarrow z_i.$
\end{algorithmoper}
\end{algorithm}

Observe that although PE is simple, it may have slow convergence because at each iteration, only {\em two} of the $N$ $z_i$'s are equalized. Conceivably, allowing {\em more} $z_i$'s to be equalized at once may speed up convergence, and this is exactly what GE does. Also notice that when $n=1$ and $P_i=1$ $\forall i\in\mathcal{V}$ for which \eqref{eq:sumPz=sumqintro} becomes the average consensus problem, PE reduces to {\em Pairwise Averaging} \cite{Tsitsiklis84} and {\em Randomized Gossip Algorithm} \cite{Boyd06}, while GE reduces to {\em Distributed Random Grouping} \cite{ChenJY06}.

\begin{remark}\label{rem:difference}
Note that existing distributed convex optimization algorithms (e.g., \cite{Rabbat05, Nedic09, ZhuM12, LuJ12b, Towfic14, Kia15}) may not be able to solve problem \eqref{eq:sumPz=sumq} over the agent network~\ref{enu:k=0}--\ref{enu:k=12rest} even though \eqref{eq:sumPz=sumq} can be viewed as an unconstrained convex quadratic program. The reason is that unlike SE, these algorithms may not be able to handle the switching of an agent's state variables from real-valued to $\#$ as the agent leaves the network and loses its memory, and from $\#$ to real-valued as the agent joins the network and interacts with others. Other main differences between SE and these algorithms include: (i) SE does not require the use of a stepsize whereas many of these algorithms do; and (ii) as mentioned above SE reduces to known algorithms in very special cases whereas most of these algorithms do not.
\end{remark}

\section{Network Connectivity}\label{sec:netwconn}

With SE, every time a subset of the $M$ agents interact and update their $z_i(k)$'s and $Q_i(k)$'s, $V(k)$ is non-increasing. While this ensures that $V(k)$ must converge, it does not ensure convergence to zero, which is desired. In fact, it is not difficult to see that for $V(k)$ to go to zero, the agent network~\ref{enu:k=0}--\ref{enu:k=12rest} must be connected in some sense. In this section, we develop a few notions of connectivity, which---unlike those in basic graph theory---are applicable to such a network.

To begin, consider a {\em hypothetical scenario}, in which $k\in\mathbb{N}$ denotes the initial time and $\ell\ge k$ the actual time. At the initial time $\ell=k$, each member $i\in\mathcal{M}(\ell)$ creates a message called message $i$, while each non-member $i\in\{1,2,\ldots,M\}-\mathcal{M}(\ell)$ has an empty memory. At each subsequent time $\ell\ge k+1$, besides action~\ref{enu:join}, each joining member $i\in\mathcal{J}(\ell)$ creates a message called message $i$ (if it has never been created) or recreates message $i$ (if it has been destroyed). Upon joining, through action~\ref{enu:inte}, all interacting members in $\mathcal{J}(\ell)\cup\mathcal{I}(\ell)\cup\mathcal{L}(\ell)$ share with one another the messages they have gathered so far. Upon interacting, besides action~\ref{enu:leav}, each leaving member $i\in\mathcal{L}(\ell)$ empties its memory and asks all staying members in $\mathcal{M}(\ell)$ to erase message $i$ from their memories, destroying message $i$. This process is then repeated indefinitely for every $\ell\ge k+1$.

For the hypothetical scenario stated above, an intriguing question is: with messages being created, shared, and destroyed as agents join the network, interact, and leave, what would be an appropriate definition of connectivity? To answer this question, recall that an undirected graph $\mathcal{G}=(\mathcal{V},\mathcal{E})$ is connected if every pair of nodes in $\mathcal{V}$ is connected by a path of edges in $\mathcal{E}$. In other words, it is not possible to partition $\mathcal{V}$ into two nonempty subsets $\mathcal{V}_1$ and $\mathcal{V}_2$ and have no paths connecting the nodes in $\mathcal{V}_1$ with those in $\mathcal{V}_2$. Motivated by this, we say that the agent network is {\em disconnected under $\mathcal{A}$ at time $k\in\mathbb{N}$} if $\mathcal{A}$ in \eqref{eq:A=FJILJIL} is such that for every $\ell\ge k$, $\mathcal{M}(\ell)$ can be partitioned into two nonempty subsets $\mathcal{M}_1(\ell)$ and $\mathcal{M}_2(\ell)$, such that all the members in $\mathcal{M}_1(\ell)$ are unaware of any messages created by those in $\mathcal{M}_2(\ell)$, and vice versa. In this definition, the phrase ``under $\mathcal{A}$ at time $k\in\mathbb{N}$'' is needed because the statement may be true for some $\mathcal{A}$ and $k$, and false for others. Likewise, the quantifier ``for every $\ell\ge k$'' is added so that being disconnected means there are {\em always} two groups of messages, which are separable.

Although it is mathematically precise, the above definition may not be readily useful in analysis because checking whether $\mathcal{M}(\ell)$ can be so partitioned for infinitely many $\ell$'s may be prohibitive. Also, if the network is {\em not} disconnected (i.e., is connected), the definition says nothing about how well-connected it is. To overcome these two limitations, let us associate with each initial time $k\in\mathbb{N}$, each subsequent time $\ell\ge k$, and each agent $i\in\{1,2,\ldots,M\}$ a set $C_i(k,\ell)\subset\mathcal{M}(\ell)$ which, roughly speaking, keeps track of the subset of members that cannot be partitioned without message crossovers. More precisely, for each $k\in\mathbb{N}$, let $C_i(k,\ell)$ $\forall i\in\{1,2,\ldots,M\}$ be initialized at $\ell=k$ to
\begin{align}
C_i(k,k)=\begin{cases}\{i\}, & \text{if $i\in\mathcal{M}(k)$},\\ \emptyset, & \text{otherwise},\end{cases}\label{eq:Cinitial}
\end{align}
and defined recursively for each $\ell\ge k+1$ as
\begin{align}
C_i(k,\ell)\!=\!\begin{cases}\displaystyle\Bigl(\bigcup_{j\in\mathcal{I}(\ell)\cup\mathcal{L}(\ell)\!\!\!\!\!\!\!\!\!\!\!\!\!\!\!\!\!\!\!\!}C_j(k,\ell-1)\cup\mathcal{J}(\ell)\Bigr)-\mathcal{L}(\ell),\\ \quad\quad\quad\text{if $\displaystyle i\!\in\!\Bigl(\bigcup_{j\in\mathcal{I}(\ell)\cup\mathcal{L}(\ell)\!\!\!\!\!\!\!\!\!\!\!\!\!\!\!\!\!\!\!\!}C_j(k,\ell\!-\!1)\!\cup\!\mathcal{J}(\ell)\Bigr)\!-\!\mathcal{L}(\ell)$},\\ \emptyset,\hfill\text{if $i\in\mathcal{L}(\ell)$},\\ C_i(k,\ell-1),\hfill\text{otherwise}.\end{cases}\label{eq:C}
\end{align}
Then, by induction on $\ell$ using \eqref{eq:Cinitial} and \eqref{eq:C}, we see that: (i) $\forall k\in\mathbb{N}$, $\forall\ell\ge k$, and $\forall i\in\{1,2,\ldots,M\}$, if $i\in\mathcal{M}(\ell)$ then $i\in C_i(k,\ell)\subset\mathcal{M}(\ell)$, otherwise $C_i(k,\ell)=\emptyset$; (ii) $\forall k\in\mathbb{N}$, $\forall\ell\ge k$, and $\forall i,j\in\mathcal{M}(\ell)$, either $C_i(k,\ell)\cap C_j(k,\ell)=\emptyset$ or $C_i(k,\ell)=C_j(k,\ell)$; and (iii) $\forall k\in\mathbb{N}$, $\forall\ell\ge k$, and $\forall i\in\mathcal{M}(\ell)$, $C_i(k,\ell)$ is the largest subset of $\mathcal{M}(\ell)$ containing agent $i$ that cannot be partitioned into two nonempty subsets, such that all the members in one are unaware of any messages from those in the other. It follows from~(i)--(iii) that the agent network is connected under $\mathcal{A}$ at time $k\in\mathbb{N}$ if and only if $\mathcal{A}$ in \eqref{eq:A=FJILJIL} is such that there exists $\ell'\ge k$ with $\ell'<\infty$, such that $C_i(k,\ell')=\mathcal{M}(\ell')$ $\forall i\in\mathcal{M}(\ell')$ (note that if such an $\ell'$ exists, then $C_i(k,\ell)=\mathcal{M}(\ell)$ $\forall i\in\mathcal{M}(\ell)$ $\forall\ell>\ell'$). This necessary and sufficient condition is more useful in analysis than the original definition (i.e., checking whether $\mathcal{M}(\ell)$ can be partitioned) because it leverages \eqref{eq:Cinitial} and \eqref{eq:C} and eliminates the need to record what messages are known to which agents at what times, which is rather cumbersome. Additionally, if the network is connected at time $k$, the smallest such $\ell'$, denoted as $\ell^*$, is a measure of how well-connected it is because $\ell^*-k$ represents the number of time instants required for the messages to become inseparable. Thus, this condition bypasses the two aforementioned limitations.

Observe that for a given $\mathcal{A}$, the network may be disconnected at certain times, and connected at others, during which it may require different number of time instants (i.e., $\ell^*-k$) for the messages to become inseparable (note that $\ell^*$ depends on $k$). To reflect these different levels of connectedness, let us introduce a function $h:\mathbb{N}\rightarrow\mathbb{N}\cup\{\infty\}$ and a constant $h^*\in\mathbb{N}\cup\{\infty\}$, defined as
\begin{align}
h(k)&=\inf D_k-k,\quad\forall k\in\mathbb{N},\label{eq:hk}\displaybreak[0]\\
h^*&=\sup_{k\in\mathbb{N}}h(k),\label{eq:hstar}
\end{align}
where the set $D_k\subset\{k,k+1,\ldots\}$ is given by
\begin{align}
D_k=\{\ell\ge k:C_i(k,\ell)\!=\!\mathcal{M}(\ell)\;\forall i\!\in\!\mathcal{M}(\ell)\},\quad\forall k\!\in\!\mathbb{N}.\label{eq:Dk}
\end{align}
With \eqref{eq:hk}--\eqref{eq:Dk}, we have $h(k)=\ell^*-k$ if the network is connected at time $k$ (due to definition of $\ell^*$), $h(k)=\infty$ otherwise (due to $D_k=\emptyset$ and $\inf\emptyset=\infty$), and $h^*<\infty$ if and only if $h$ is bounded. Hence, the smaller $h(k)$ and $h^*$, the better the ``instantaneous'' and ``worst-case'' connectedness, respectively. Putting all of the above together, we arrive at the following formal definition:

\begin{definition}\label{def:netwconn}
The agent network modeled by~\ref{enu:k=0}--\ref{enu:k=12rest} is said to be {\em connected under $\mathcal{A}$ at time $k\in\mathbb{N}$} if $h(k)<\infty$. It is said to be {\em connected under $\mathcal{A}$} if $h(k)<\infty$ $\forall k\in\mathbb{N}$, and {\em uniformly connected under $\mathcal{A}$} if $h^*<\infty$.
\end{definition}

To illustrate the above ideas, consider Figure~\ref{fig:netwconn}, which shows a $6$-agent network at some time $k$ and its evolution until time $k+4$. In this figure, an agent $i$ is a member at time $\ell$ if and only if it is enclosed by a black dashed curve (e.g., agent $6$ is not a member at time $k$). Moreover, if an agent $i$ at time $\ell$ is enclosed by a gray solid curve, then $C_i(k,\ell)$ is the set of agents enclosed by the same curve (e.g., $C_1(k,k+1)=\{1\}$, $C_2(k,k+1)=\{2\}$, $C_3(k,k+1)=C_5(k,k+1)=\{3,5\}$, and $C_4(k,k+1)=\{4\}$). Otherwise, $C_i(k,\ell)$ is empty (e.g., $C_6(k,k+1)=\emptyset$). Note that the black dashed curve at time $k$ is arbitrarily selected, whereas those at subsequent times are due to \eqref{eq:M=McupJ-L}. Similarly, the gray solid curves at time $k$ are due to \eqref{eq:Cinitial}, whereas those at subsequent times are due to \eqref{eq:C}. Examining these curves along with \eqref{eq:Dk}, we deduce that the set $D_k$ does not contain $k$, $k+1$, $k+2$, and $k+3$ but contains $k+4$. From \eqref{eq:hk} and Definition~\ref{def:netwconn}, we conclude that $h(k)=4$ and, hence, the network is connected under $\mathcal{A}$ at time $k$.

\begin{figure*}[!t]
\centering\begin{texdraw}\input{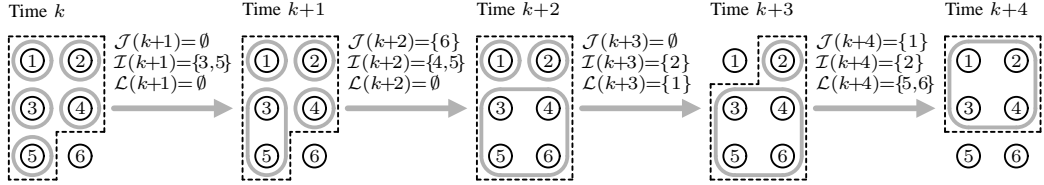}\end{texdraw}
\caption{An illustration of network connectivity in a $6$-agent network. Each black dashed curve represents $\mathcal{M}(\ell)$ and each gray solid curve represents $C_i(k,\ell)$.}
\label{fig:netwconn}
\end{figure*}

The following examples further illustrate Definition~\ref{def:netwconn}:

\begin{example}\label{exa:disconn}
Consider the agent network~\ref{enu:k=0}--\ref{enu:k=12rest} and suppose $M=3$. Let $\mathcal{F}=\{1,2\}$ and $(\mathcal{J}(k),\mathcal{I}(k),\mathcal{L}(k))$ be equal to $(\{3\},\{1\},\emptyset)$ if $(k-1)/6\in\mathbb{N}$, to $(\emptyset,\{3\},\{1\})$ if $(k-2)/6\in\mathbb{N}$, to $(\{1\},\{2\},\emptyset)$ if $(k-3)/6\in\mathbb{N}$, to $(\emptyset,\{1\},\{2\})$ if $(k-4)/6\in\mathbb{N}$, to $(\{2\},\{3\},\emptyset)$ if $(k-5)/6\in\mathbb{N}$, and to $(\emptyset,\{2\},\{3\})$ if $(k-6)/6\in\mathbb{N}$, thereby defining $\mathcal{A}$ in \eqref{eq:A=FJILJIL}. Examining $\mathcal{A}$, we see two groups of messages being passed around the agents, but {\em never} getting a chance to ``mix.'' Thus, we expect the network to be disconnected under $\mathcal{A}$ at all times. Indeed, applying \eqref{eq:Cinitial}, \eqref{eq:C}, \eqref{eq:hk}, \eqref{eq:Dk}, and Definition~\ref{def:netwconn} yields $h(k)=\infty$ $\forall k\in\mathbb{N}$, confirming the expectation.\hfill$\square$
\end{example}

\begin{example}\label{exa:unifconn}
Reconsider the agent network in Example~\ref{exa:disconn} but let $(\mathcal{J}(k),\mathcal{I}(k),\mathcal{L}(k))$ be equal to $(\{3\},\{1\},\emptyset)$ if $(k-1)/6\in\mathbb{N}$, to $(\emptyset,\{2\},\{1\})$ if $(k-2)/6\in\mathbb{N}$, to $(\{1\},\{2\},\emptyset)$ if $(k-3)/6\in\mathbb{N}$, to $(\emptyset,\{3\},\{2\})$ if $(k-4)/6\in\mathbb{N}$, to $(\{2\},\{3\},\emptyset)$ if $(k-5)/6\in\mathbb{N}$, and to $(\emptyset,\{1\},\{3\})$ if $(k-6)/6\in\mathbb{N}$. Observe that unlike the $\mathcal{A}$ in Example~\ref{exa:disconn}, the $\mathcal{A}$ here causes the messages to quickly become inseparable no matter the initial time. Hence, the network is expected to not only be connected, but uniformly so, under $\mathcal{A}$. It follows from \eqref{eq:hk} that $h(k)=2$ if $k$ is even and $h(k)=3$ if $k$ is odd, from \eqref{eq:hstar} that $h^*=3$, and from Definition~\ref{def:netwconn} that the network is indeed uniformly connected.\hfill$\square$
\end{example}

\begin{example}\label{exa:notunifconn}
Reconsider the agent network in Example~\ref{exa:disconn} but let $\mathcal{F}=\{1,2,3\}$ and $(\mathcal{J}(k),\mathcal{I}(k),\mathcal{L}(k))$ be equal to $(\emptyset,\{1,2\},\emptyset)$ if $k\in\{\ell(\ell+1)/2:\ell\in\mathbb{P}\}$ and to $(\emptyset,\{2,3\},\emptyset)$ otherwise. Notice that although agent $2$ takes turn to interact with agents $1$ and $3$, its interaction with agent $1$ becomes less and less frequent, as if the network is gradually losing its connectivity. Therefore, the network is expected to be connected, but not uniformly so, under $\mathcal{A}$. Indeed, it is connected because $h(0)=2$ and $h(k)\le\ell+1<\infty$ $\forall\ell\in\mathbb{P}$ $\forall k\in[\ell(\ell+1)/2,(\ell+1)(\ell+2)/2-1]$. It is not uniformly connected because $h(\ell(\ell+1)/2)=\ell+1$ $\forall\ell\in\mathbb{P}$ so that $h^*=\infty$.\hfill$\square$
\end{example}

Finally, it might be of interest to see how Definition~\ref{def:netwconn} is related to existing definitions of connectivity in the literature. The following proposition sheds light on this question, showing that when there are no membership changes, the connectivity of the agent network under $\mathcal{A}$ is equivalent to the connectivity of an infinite interaction graph first introduced in \cite{Tsitsiklis84}, so that the former is a generalization of the latter:

\begin{proposition}\label{pro:netwconn}
If $\mathcal{J}(k)=\mathcal{L}(k)=\emptyset$ $\forall k\in\mathbb{P}$, then the agent network~\ref{enu:k=0}--\ref{enu:k=12rest} is connected under $\mathcal{A}$ if and only if the graph $(\mathcal{F},\mathcal{E}_\infty)$ is connected, where
\begin{align}
\mathcal{E}_\infty=\{\{i,j\}\subset\mathcal{F}:\{i,j\}\subset\mathcal{I}(k)\;\text{for infinitely many}\;k\in\mathbb{P}\}.\label{eq:Einfty}
\end{align}
\end{proposition}

\section{Boundedness and Convergence}\label{sec:bounconv}

In this section, we analyze the boundedness, asymptotic convergence, and exponential convergence of SE and derive a bound on its convergence rate. To streamline the presentation of the results, we defer their proofs to the Appendix. Moreover, we let $\beta>0$ denote the spectral radius of $\sum_{i\in\mathcal{M}(0)}P_i$ and introduce the following definition:

\begin{definition}\label{def:upd}
The sequence $\{Q_i(k)\}_{k\in\mathbb{N},i\in\mathcal{M}(k)}$ produced by SE is said to be {\em uniformly positive definite under $\mathcal{A}$} if $\exists\alpha>0$ such that $\forall k\in\mathbb{N}$, $\forall i\in\mathcal{M}(k)$, $Q_i(k)-\alpha I\in\mathbb{S}_+^n$.
\end{definition}

Although the initial values $Q_i(0)$'s depend on the observations $P_i$'s via \eqref{eq:Qinitial}, it can be verified that the uniform positive definiteness of $\{Q_i(k)\}_{k\in\mathbb{N},i\in\mathcal{M}(k)}$ depends only on the agent actions $\mathcal{A}$ and not on the $P_i$'s nor the $q_i$'s.

We first give a sufficient condition on SE's boundedness:

\begin{theorem}\label{thm:SEboun}
Consider the agent network~\ref{enu:k=0}--\ref{enu:k=12rest} and the use of SE. Let $\mathcal{A}$ be given. Then, $Q_i(k)$ is bounded as follows:
\begin{align}
Q_i(k)\le\beta I,\quad\forall k\in\mathbb{N},\;\forall i\in\mathcal{M}(k).\label{eq:Qbound}
\end{align}
If, in addition, the sequence $\{Q_i(k)\}_{k\in\mathbb{N},i\in\mathcal{M}(k)}$ is uniformly positive definite under $\mathcal{A}$, then $z_i(k)$ is bounded as follows:
\begin{align}
\|z_i(k)-z\|^2\le\frac{V(k)}{\alpha}\le\frac{V(0)}{\alpha},\quad\forall k\in\mathbb{N},\;\forall i\!\in\!\mathcal{M}(k),\label{eq:zbound}
\end{align}
where $\alpha$ is any positive number satisfying $Q_i(k)-\alpha I\in\mathbb{S}_+^n$ $\forall k\in\mathbb{N}$ $\forall i\in\mathcal{M}(k)$.
\end{theorem}

Theorem~\ref{thm:SEboun} implies that all the $Q_i(k)$'s of the members are unconditionally bounded from above by $\beta$, irrespective of the agent actions $\mathcal{A}$. In addition, if they turn out to be bounded from below by some $\alpha>0$, then all the $z_i(k)$'s of the members are guaranteed to stay within a ball centered at the solution $z$, whose radius $\sqrt{V(k)/\alpha}$ decreases over time.

In general, given $\mathcal{A}$, it is not easy to check whether the resulting sequence $\{Q_i(k)\}_{k\in\mathbb{N},i\in\mathcal{M}(k)}$ is uniformly positive definite under $\mathcal{A}$. However, if $\mathcal{A}$ happens to be such that every agent joins and leaves the network arbitrarily but {\em finitely} many times---a rather mild condition that is often satisfied in practice---then the uniform positive definiteness of $\{Q_i(k)\}_{k\in\mathbb{N},i\in\mathcal{M}(k)}$ can be immediately verified. The definition and corollary to Theorem~\ref{thm:SEboun} below formalize this claim:

\begin{definition}\label{def:ultistat}
The membership dynamics \eqref{eq:M=McupJ-L} of the agent network~\ref{enu:k=0}--\ref{enu:k=12rest} are said to be {\em ultimately static under $\mathcal{A}$} if $\exists k\in\mathbb{N}$ such that $\forall\ell>k$, $\mathcal{M}(\ell)=\mathcal{M}(k)$, i.e., $\mathcal{J}(\ell)=\mathcal{L}(\ell)=\emptyset$.
\end{definition}

\begin{corollary}\label{cor:SEboun}
If the membership dynamics \eqref{eq:M=McupJ-L} are ultimately static under $\mathcal{A}$, then $Q_i(k)$ and $z_i(k)$ are bounded as in \eqref{eq:Qbound} and \eqref{eq:zbound} for some $\alpha>0$.
\end{corollary}

In Theorem~\ref{thm:SEboun} and Corollary~\ref{cor:SEboun}, the network is not assumed to be connected since such an assumption is not needed for the boundedness of SE. For convergence, however, this assumption is crucial. The following lemma, which makes use of this assumption, is a key step toward establishing both the asymptotic and exponential convergence of SE:

\begin{lemma}\label{lem:SEVexpdecr}
Consider the agent network~\ref{enu:k=0}--\ref{enu:k=12rest} and the use of SE. Let $\mathcal{A}$ be given. Suppose the agent network is connected under $\mathcal{A}$ at time $k\in\mathbb{N}$, so that $h(k)<\infty$. Then,
\begin{align}
V(k+h(k))\le\frac{\gamma(k)}{\gamma(k)+1}V(k),\label{eq:V<=fV}
\end{align}
where $\alpha$ is any positive number satisfying $Q_i(\ell)-\alpha I\in\mathbb{S}_+^n$ $\forall\ell\in[k,k+h(k)]$ $\forall i\in\mathcal{M}(\ell)$, and
\begin{align*}
\gamma(k)=M\min\{M^{h(k)+1},M!\}\Bigl(\frac{4\beta}{\alpha}\Bigr)^{\min\{h(k),M-1\}}>0.
\end{align*}
\end{lemma}

Lemma~\ref{lem:SEVexpdecr} asserts that as long as the agent network is connected under $\mathcal{A}$ at some time $k$, the value of $V$ must strictly decrease from $V(k)$ at time $k$ to $V(k+h(k))$ at time $k+h(k)$, by a factor that can be explicitly calculated in \eqref{eq:V<=fV}. This result suggests that the better the ``instantaneous'' connectedness (i.e., the smaller $h(k)$), the faster the value of $V$ drops, which makes intuitive sense. However, even if $V(k)$ decreases asymptotically to zero as $k\rightarrow\infty$, it does not\footnote{For instance, let $M=3$, $\mathcal{F}=\{1,2\}$, and $(\mathcal{J}(k),\mathcal{I}(k),\mathcal{L}(k))$ be equal to $(\{3\},\{1\},\emptyset)$ if $k$ is odd and to $(\emptyset,\{2\},\{3\})$ if $k$ is even, thus defining $\mathcal{A}$ in \eqref{eq:A=FJILJIL}. Also, let $P_1=P_2=1$, $q_1=1$, and $q_2=2$, so that $z=1.5$ from \eqref{eq:sumPz=sumq}. With this $\mathcal{A}$, agent $3$ repeatedly does the following: joins the network, interacts with agent $1$ upon joining, leaves the network subsequently, and interacts with agent $2$ prior to leaving. Hence, the network is connected under $\mathcal{A}$. Moreover, $\forall k\in\mathbb{N}$, $Q_1(k)=(\frac{1}{2})^{\lceil\frac{k}{2}\rceil}$, $Q_2(k)=2-(\frac{1}{2})^{\lfloor\frac{k}{2}\rfloor}$, $Q_3(k)=(\frac{1}{2})^{\lceil\frac{k}{2}\rceil}$ if $k$ is odd, $Q_3(k)=\#$ if $k$ is even, $z_1(k)=1$, $z_2(k)=(3-(\frac{1}{2})^{\lfloor\frac{k}{2}\rfloor})/(2-(\frac{1}{2})^{\lfloor\frac{k}{2}\rfloor})$, $z_3(k)=1$ if $k$ is odd, and $z_3(k)=\#$ if $k$ is even. Thus, we have $\lim_{k\rightarrow\infty}V(k)=0$ but $\lim_{k\rightarrow\infty}z_1(k)=1\ne z.$} necessarily imply that all the $z_i(k)$'s of SE would asymptotically converge to $z$ because some of the $Q_i(k)$'s might be losing their positive definiteness as $k\rightarrow\infty$. This phenomenon suggests that network connectivity {\em and} the uniform positive definiteness of $\{Q_i(k)\}_{k\in\mathbb{N},i\in\mathcal{M}(k)}$ together might be all that are needed to establish the asymptotic convergence of SE. The following theorem shows that this is indeed the case:

\begin{theorem}\label{thm:SEasymconv}
Consider the agent network~\ref{enu:k=0}--\ref{enu:k=12rest} and the use of SE. Let $\mathcal{A}$ be given. Suppose the agent network is connected under $\mathcal{A}$ and the sequence $\{Q_i(k)\}_{k\in\mathbb{N},i\in\mathcal{M}(k)}$ is uniformly positive definite under $\mathcal{A}$. Then, $z_i(k)$ asymptotically converges to the solution $z$, i.e.,
\begin{align}
&\forall\varepsilon>0,\;\exists k\in\mathbb{N}\;\text{such that}\nonumber\displaybreak[0]\\
&\quad\forall\ell\ge k,\;\forall i\in\mathcal{M}(\ell),\;\|z_i(\ell)-z\|<\varepsilon.\label{eq:limz=z}
\end{align}
\end{theorem}

\begin{corollary}\label{cor:SEasymconv}
If the agent network is connected under $\mathcal{A}$ and the membership dynamics \eqref{eq:M=McupJ-L} are ultimately static under $\mathcal{A}$, then \eqref{eq:limz=z} holds.
\end{corollary}

\begin{proof}
The proof is an immediate consequence of Theorem~\ref{thm:SEasymconv} and the proof of Corollary~\ref{cor:SEboun}.
\end{proof}

Note that the conclusion of Theorem~\ref{thm:SEasymconv} is written as \eqref{eq:limz=z} instead of ``$\lim_{k\rightarrow\infty}z_i(k)=z$'' because the former excludes cases where $z_i(k)=\#$, while the latter does not and, thus, is not well-defined. More important, with Theorem~\ref{thm:SEasymconv} and Corollary~\ref{cor:SEasymconv}, we achieve the paper's objective of developing a distributed asynchronous algorithm SE that asymptotically solves \eqref{eq:sumPz=sumq} over the agent network~\ref{enu:k=0}--\ref{enu:k=12rest}, while possessing properties~\ref{enu:exog}--\ref{enu:memo} stated in Section~\ref{sec:netwmodeprobform}.

Finally, we provide a sufficient condition on the exponential convergence of SE and derive a bound on its convergence rate, in terms of $h^*$. Since $h^*=0$ is a trivial case (that corresponds to $\mathcal{M}(k)$ containing exactly one and the same agent $i\in\{1,2,\ldots,M\}$ with $z_i(k)=z$ $\forall k\in\mathbb{N}$), below it is assumed that $h^*>0$:

\begin{theorem}\label{thm:SEconvrate}
Consider the agent network~\ref{enu:k=0}--\ref{enu:k=12rest} and the use of SE. Let $\mathcal{A}$ be given. Suppose the agent network is uniformly connected under $\mathcal{A}$ with $h^*>0$ and the sequence $\{Q_i(k)\}_{k\in\mathbb{N},i\in\mathcal{M}(k)}$ is uniformly positive definite under $\mathcal{A}$. Then,
\begin{align}
&V(k)\le V(0)\Bigl(\frac{\gamma^*}{\gamma^*+1}\Bigr)^{\lfloor\frac{k}{h^*}\rfloor},\quad\forall k\in\mathbb{N},\label{eq:V<=V0f}\displaybreak[0]\\
&\|z_i(k)-z\|^2\le\frac{V(0)}{\alpha}\Bigl(\frac{\gamma^*}{\gamma^*+1}\Bigr)^{\lfloor\frac{k}{h^*}\rfloor},\quad\forall k\in\mathbb{N},\;\forall i\in\mathcal{M}(k),\label{eq:|z-z|<=V0/alphaf}
\end{align}
where $\alpha$ is any positive number satisfying $Q_i(k)-\alpha I\in\mathbb{S}_+^n$ $\forall k\in\mathbb{N}$ $\forall i\in\mathcal{M}(k)$, and
\begin{align*}
\gamma^*=M\min\{M^{h^*+1},M!\}\Bigl(\frac{4\beta}{\alpha}\Bigr)^{\min\{h^*,M-1\}}>0.
\end{align*}
\end{theorem}

\begin{corollary}\label{cor:SEconvrate}
If the agent network is uniformly connected under $\mathcal{A}$ with $h^*>0$ and the membership dynamics \eqref{eq:M=McupJ-L} are ultimately static under $\mathcal{A}$, then \eqref{eq:V<=V0f} and \eqref{eq:|z-z|<=V0/alphaf} hold for some $\alpha>0$.
\end{corollary}

\begin{proof}
The proof follows immediately from Theorem~\ref{thm:SEconvrate} and the proof of Corollary~\ref{cor:SEboun}.
\end{proof}

Observe from Theorems~\ref{thm:SEasymconv} and~\ref{thm:SEconvrate} (or Corollaries~\ref{cor:SEasymconv} and~\ref{cor:SEconvrate}) that {\em connectivity} helps ensure {\em asymptotic convergence}, while {\em uniform connectivity} helps ensure {\em exponential convergence}.

\section{Simulation Studies}\label{sec:simustud}

In this section, we complement the above analysis with simulation. Section~\ref{ssec:illuSEagennetw} illustrates the behavior of SE in a volatile agent network. Section~\ref{ssec:compPEGEexisalgowirenetw} compares the performance of PE and GE with a few existing algorithms in multi-hop wireless networks with fixed vertex sets.

\subsection{Illustration of SE in an Agent Network}\label{ssec:illuSEagennetw}

In this subsection, we simulate SE in an agent network described by~\ref{enu:k=0}--\ref{enu:k=12rest} with the following settings: $M=100$; $\mathcal{F}=\{1,2,\ldots,50\}$; $n=4$; for each $i\in\mathcal{F}$, $P_i\in\mathbb{S}_+^n$ and $q_i\in\mathbb{R}^n$ are randomly generated; and for each $k\in\mathbb{P}$, the sets $\mathcal{J}(k)$, $\mathcal{I}(k)$, and $\mathcal{L}(k)$ are random subsets of the sets $\{1,2,\ldots,M\}-\mathcal{M}(k-1)$, $\mathcal{M}(k-1)$, and $\mathcal{M}(k-1)$, respectively, such that $\mathcal{I}(k)\cap\mathcal{L}(k)=\emptyset$, $\mathcal{I}(k)\ne\emptyset$, and $\mathcal{L}(k)\subsetneq\mathcal{M}(k-1)$ according to~\ref{enu:k=12}. Note that with these settings, the agent network is volatile with random, unpredictable member interactions and membership dynamics. Thus, the behavior of SE in such a network is indicative of its effectiveness.

Figure~\ref{fig:illuexam} depicts the simulation results. The top subplot of Figure~\ref{fig:illuexam} shows the number of members $|\mathcal{M}(k)|$ as a function of time $k$. The middle subplot shows the actions taken by two selected agents, agent $1$ and agent $51$, at each time $k$, where a total of five actions are possible as labeled on the vertical axis, and only the actions of two agents are shown to avoid clogging the plot. Also, the actions labeled ``$i\in\mathcal{M}(k)$ but idle'' and ``$i\notin\mathcal{M}(k)$ but idle'' are abbreviations for $i\in\mathcal{M}(k-1)-(\mathcal{I}(k)\cup\mathcal{L}(k))$ and $i\in\{1,2,\ldots,M\}-(\mathcal{M}(k-1)\cup\mathcal{J}(k))$, respectively. Lastly, the bottom subplot shows, on a logarithmic scale and as functions of time $k$, the maximum estimation error $\max_{i\in\mathcal{M}(k)}\|z_i(k)-z\|$ among the members in $\mathcal{M}(k)$, the minimum such error $\min_{i\in\mathcal{M}(k)}\|z_i(k)-z\|$, the estimation error $\|z_1(k)-z\|$ of agent $1$ whenever it is a member, and the estimation error $\|z_{51}(k)-z\|$ of agent $51$ whenever it is a member. Observe from the figure that, despite the rapidly fluctuating number of members, and despite the randomly generated actions of agents that include numerous membership changes, all the estimates $z_i(k)$'s gradually approach the unknown solution $z$, demonstrating the effectiveness of SE.

\begin{figure}[!t]
\centering\includegraphics[width=0.85\linewidth]{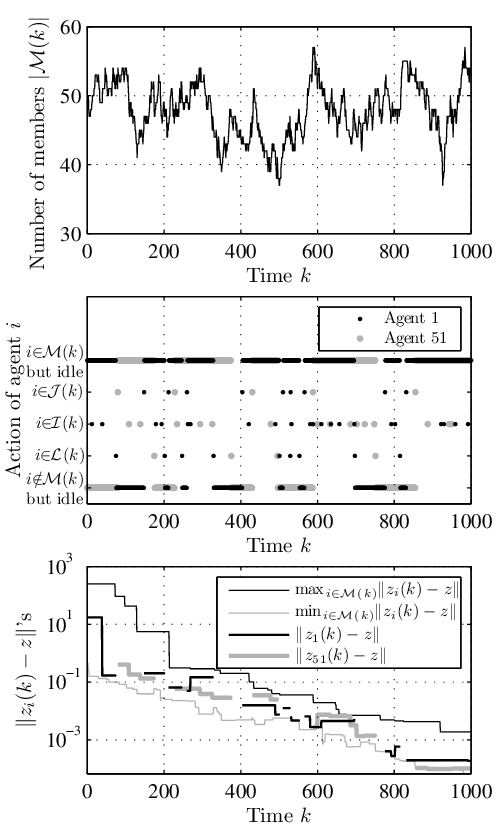}
\caption{Behavior of SE in a volatile agent network with random, unpredictable member interactions and membership dynamics.}
\label{fig:illuexam}
\end{figure}

\subsection{Comparison of PE and GE with Existing Algorithms in Wireless Networks}\label{ssec:compPEGEexisalgowirenetw}

In this subsection, we compare through simulation the performance of five algorithms---namely, PE and GE from Algorithms~\ref{alg:PE} and~\ref{alg:GE}, the two average-consensus-based algorithms from \cite{XiaoL05, XiaoL06} called {\em Maximum-Degree Weights} (MDW) and {\em Metropolis Weights} (MW), and flooding---in solving problem \eqref{eq:sumPz=sumqintro} of different sizes, over multi-hop wireless networks modeled by random geometric graphs of different sizes and densities. The simulation settings are as follows: to methodically evaluate the algorithm performance, we let the simulation be governed by three parameters---the number of nodes $N$ that represents {\em network sizes}, the average number of neighbors $\frac{2L}{N}$ that represents {\em network densities} (the meaning of $\frac{2L}{N}$ will be clear shortly), and the number of dimensions $n$ that represents {\em problem sizes}---and write them as a $3$-tuple $(N,\frac{2L}{N},n)$. To understand their individual impact, we vary these parameters one at a time, choosing the values of $(N,\frac{2L}{N},n)$ as:
\begin{enumerate}
\renewcommand{\theenumi}{S\arabic{enumi}}\itemsep-\parsep
\item $(\mathit{50},20,4),(\mathit{100},20,4),\ldots,(\mathit{500},20,4)$;\label{enu:netwsize}
\item $(200,\mathit{10},4),(200,\mathit{20},4),\ldots,(200,\mathit{100},4)$; and\label{enu:netwdens}
\item $(200,20,\mathit{2}),(200,20,\mathit{4}),\ldots,(200,20,\mathit{20})$.\label{enu:probsize}
\end{enumerate}

For each value of $(N,\frac{2L}{N},n)$ in~\ref{enu:netwsize}--\ref{enu:probsize}, we consider $50$ random scenarios. For each scenario, we generate a wireless network with $N$ nodes and $L$ edges by randomly and equiprobably placing $N$ nodes on a unit square in $\mathbb{R}^2$ and gradually increasing the one-hop radius until the number of edges is $L$ or, equivalently, the average number of neighbors is $\frac{2L}{N}$ (this explains the meaning of $\frac{2L}{N}$). If the resulting network is not connected, it is discarded and the preceding process is repeated. We also generate an instance of problem \eqref{eq:sumPz=sumqintro} with $n$ dimensions by factoring each $P_i\in\mathbb{S}_+^n$ as $P_i=X_i^TX_i$ and letting both $X_i\in\mathbb{R}^{n\times n}$ and $q_i\in\mathbb{R}^n$ have random entries drawn independently from the standard normal distribution. Subsequently, we simulate PE, GE, MDW, and MW and let the gossiping pair in Step~3 of PE, as well as the interacting group in Step~3 of GE, be randomly and equiprobably chosen. We then count the number of real-number transmissions needed for each algorithm to converge (including initialization overhead, if any), where the convergence criterion is $\max_{i\in\mathcal{V}}\|z_i(k)-z\|<0.005$. To count such numbers, we use the fact that PE, GE, MDW, and MW require, respectively, $\frac{n(n+1)}{2}N$, $\frac{n(n+1)}{2}N$, $0$, and $0$ real-number transmissions to initialize and $2n$, $n(|\mathcal{N}_i|+1)$, $(\frac{n(n+1)}{2}+n)N$, and $(\frac{n(n+1)}{2}+n)N$ real-number transmissions per iteration (in the case of GE, per iteration initiated by node $i$). Finally, for each value of $(N,\frac{2L}{N},n)$ in~\ref{enu:netwsize}--\ref{enu:probsize} and for each algorithm, we average over the $50$ scenarios and record the resulting number needed to converge. As a benchmark, we also record the number needed by flooding to exactly solve \eqref{eq:sumPz=sumqintro} (i.e., $(\frac{n(n+1)}{2}+n)N^2$).

Figure~\ref{fig:nrtc} displays the simulation results, showing in its subplots~(a), (b), and~(c) the number of real-number transmissions needed as a function of the values of $(N,\frac{2L}{N},n)$ in~\ref{enu:netwsize}, \ref{enu:netwdens}, and~\ref{enu:probsize}, respectively. Notice from the figure that:
\begin{itemize}
\itemsep-\parsep
\item Generally, the larger the network size, or the lower the network density, or the larger the problem size, the higher the number needed. One exception to this trend is flooding in subplot~(b), which is expected since its number depends only on $N$ and $n$ and not on $L$.
\item Among the five algorithms, MDW has, on average, the worst {\em bandwidth/\linebreak[0]energy efficiency}, requiring by far the most real-number transmissions to converge. Nonetheless, MDW does outperform flooding when the network is sufficiently dense.
\item PE is not as efficient as MW in subplots~(a) and~(b). However, it becomes more efficient than MW when the problem size is sufficiently large, in subplot~(c). This is likely due to PE being $O(n)$ and MW being $O(n^2)$ in the number of real-number transmissions per iteration.
\item Among the five algorithms, GE has the best bandwidth/\linebreak[0]energy efficiency and scalability with respect to $N$ and $n$. Indeed, GE is at least $2.5$ times and up to $8$ times more efficient than the next best algorithm---be it MW or PE---in all the values of $(N,\frac{2L}{N},n)$ considered.
\end{itemize}

\begin{figure*}[!t]
\centering\subfigure[Varying network sizes.]{\includegraphics[width=0.3\textwidth]{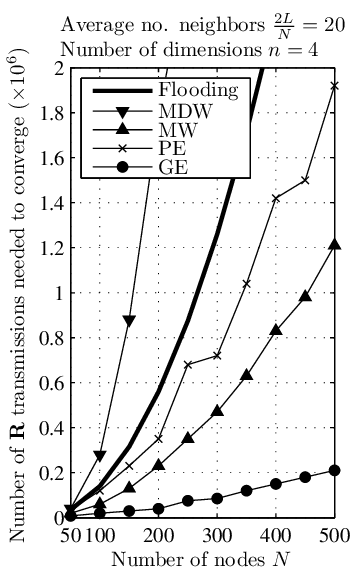}}\quad\subfigure[Varying network densities.]{\includegraphics[width=0.3\textwidth]{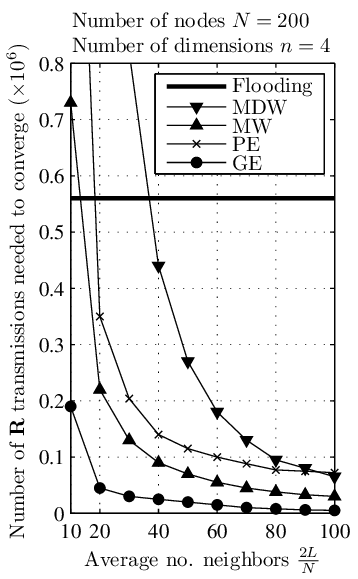}}\quad\subfigure[Varying problem sizes.]{\includegraphics[width=0.3\textwidth]{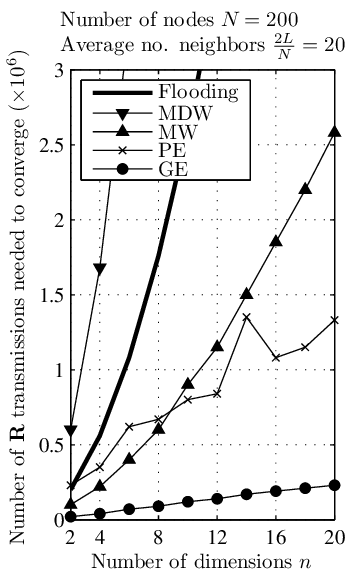}}
\caption{Performance of PE, GE, MDW, MW, and flooding in multi-hop wireless networks with varying network sizes, network densities, and problem sizes.}
\label{fig:nrtc}
\end{figure*}

\section{Conclusion}\label{sec:conc}

In this paper, we have developed SE, a distributed asynchronous algorithm for solving symmetric positive definite systems of linear equations over agent networks with arbitrary member interactions and membership dynamics. To facilitate the development, we have introduced a time-varying Lyapunov-like function and a generalized concept of network connectivity. Based on these entities, we have derived sufficient conditions for ensuring the boundedness, asymptotic convergence, and exponential convergence of SE, as well as a bound on its convergence rate. We have also shown that SE reduces to known algorithms in very special cases. Finally, we have demonstrated through extensive simulation the effectiveness and efficiency of SE in a variety of settings.

\appendix

Throughout the Appendix, for any $x\in\mathbb{R}^n$ and any $P\in\mathbb{S}_+^n$, we write $x^Tx$ and $x^TPx$ as $\|x\|^2$ and $\|x\|_P^2$, respectively. Also, for any $k\in\mathbb{N}$ and any nonempty $X\subset\mathcal{M}(k)$, we let
\begin{align}
z_X^k=\Bigl(\sum_{i\in X}Q_i(k)\Bigr)^{-1}\sum_{i\in X}Q_i(k)z_i(k),\label{eq:zxk}
\end{align}
so that from \eqref{eq:zopti},
\begin{align}
z_i(k)=z_{\mathcal{I}(k)\cup\mathcal{L}(k)}^{k-1},\quad\forall k\in\mathbb{P},\;\forall i\in\mathcal{J}(k)\cup\mathcal{I}(k).\label{eq:zi=zIL}
\end{align}

\section*{Proof of Lemma~\ref{lem:optnonincr}}

Let $\mathcal{A}$ and $k\in\mathbb{P}$ be given. To prove the first statement, pick any $Q_i(k)\in\mathbb{S}_+^n$ $\forall i\in\mathcal{J}(k)\cup\mathcal{I}(k)$ satisfying \eqref{eq:sumQ=sumQ} and consider the following constrained convex optimization problem:
\begin{align}
\begin{array}{cl}\displaystyle\operatornamewithlimits{minimize}_{(z_i(k))_{i\in\mathcal{J}(k)\cup\mathcal{I}(k)}} & \displaystyle\sum_{i\in\mathcal{J}(k)\cup\mathcal{I}(k)}\!\!\!\!z_i(k)^TQ_i(k)z_i(k)\\ \operatorname{subject\,to} & \eqref{eq:sumQz=sumQz}.\end{array}\label{eq:minzsumzQz}
\end{align}
By forming the Lagrangian of problem \eqref{eq:minzsumzQz} and setting its gradient to zero, we see that problem \eqref{eq:minzsumzQz} has a unique solution $(z_i(k))_{i\in\mathcal{J}(k)\cup\mathcal{I}(k)}$ given by \eqref{eq:zopti}. Moreover, by substituting \eqref{eq:zopti} into the objective function and using \eqref{eq:sumQ=sumQ}, we see that the optimal value of problem \eqref{eq:minzsumzQz} depends only on $(z_i(k-1),Q_i(k-1))_{i\in\mathcal{I}(k)\cup\mathcal{L}(k)}$ and not on the arbitrary $(Q_i(k))_{i\in\mathcal{J}(k)\cup\mathcal{I}(k)}$. Hence, problem \eqref{eq:minsumzQz} has a nonempty, convex set of solutions given by $\{(z_i(k),Q_i(k))_{i\in\mathcal{J}(k)\cup\mathcal{I}(k)}:\text{\eqref{eq:zopti} and \eqref{eq:sumQ=sumQ} hold}\}$, i.e., the first statement is true. For the second statement, note from \eqref{eq:V-V}, \eqref{eq:sumQz=sumQz}, \eqref{eq:sumQ=sumQ}, and \eqref{eq:zi=zIL} that
\begin{align}
&V(k)-V(k-1)\nonumber\displaybreak[0]\\
=&\sum_{i\in\mathcal{J}(k)\cup\mathcal{I}(k)\!\!\!\!\!\!\!\!\!\!\!\!\!\!\!\!\!\!\!\!\!\!}z_i(k)^TQ_i(k)z_i(k)-\sum_{i\in\mathcal{I}(k)\cup\mathcal{L}(k)\!\!\!\!\!\!\!\!\!\!\!\!\!\!\!\!\!\!\!\!\!\!}z_i(k\!-\!1)^TQ_i(k\!-\!1)z_i(k\!-\!1)\nonumber\displaybreak[0]\\
=&-\!\Bigl(\sum_{i\in\mathcal{I}(k)\cup\mathcal{L}(k)\!\!\!\!\!\!\!\!\!\!\!\!\!\!\!\!\!\!\!\!\!\!}z_i(k-1)^TQ_i(k-1)z_i(k-1)\!+\!\sum_{i\in\mathcal{J}(k)\cup\mathcal{I}(k)\!\!\!\!\!\!\!\!\!\!\!\!\!\!\!\!\!\!\!\!\!\!}(z_{\mathcal{I}(k)\cup\mathcal{L}(k)}^{k-1})^T\nonumber\displaybreak[0]\\
&\times Q_i(k)z_{\mathcal{I}(k)\cup\mathcal{L}(k)}^{k-1}-2\sum_{i\in\mathcal{J}(k)\cup\mathcal{I}(k)\!\!\!\!\!\!\!\!\!\!\!\!\!\!\!\!\!\!\!\!\!\!}z_i(k)^TQ_i(k)z_{\mathcal{I}(k)\cup\mathcal{L}(k)}^{k-1}\Bigr)\nonumber\displaybreak[0]\\
=&-\!\Bigl(\sum_{i\in\mathcal{I}(k)\cup\mathcal{L}(k)\!\!\!\!\!\!\!\!\!\!\!\!\!\!\!\!\!\!\!\!\!\!}z_i(k-1)^TQ_i(k-1)z_i(k-1)\!+\!\sum_{i\in\mathcal{I}(k)\cup\mathcal{L}(k)\!\!\!\!\!\!\!\!\!\!\!\!\!\!\!\!\!\!\!\!\!\!}(z_{\mathcal{I}(k)\cup\mathcal{L}(k)}^{k-1})^T\nonumber\displaybreak[0]\\
&\times Q_i(k\!\!-\!\!1)z_{\mathcal{I}(k)\cup\mathcal{L}(k)}^{k-1}\!\!-\!\!2\sum_{i\in\mathcal{I}(k)\cup\mathcal{L}(k)\!\!\!\!\!\!\!\!\!\!\!\!\!\!\!\!\!\!\!\!\!\!}z_i(k\!\!-\!\!1)^TQ_i(k\!\!-\!\!1)z_{\mathcal{I}(k)\cup\mathcal{L}(k)}^{k-1}\Bigr)\nonumber\displaybreak[0]\\
=&-\!\!\sum_{i\in\mathcal{I}(k)\cup\mathcal{L}(k)\!\!\!\!\!\!\!\!\!\!\!\!\!\!\!\!\!\!\!\!\!\!}(z_i(k\!\!-\!\!1)\!\!-\!\!z_{\mathcal{I}(k)\cup\mathcal{L}(k)}^{k-1})^TQ_i(k\!\!-\!\!1)(z_i(k\!\!-\!\!1)\!\!-\!\!z_{\mathcal{I}(k)\cup\mathcal{L}(k)}^{k-1}).\label{eq:VV=sumzzQzz}
\end{align}
Since the right-hand side of \eqref{eq:VV=sumzzQzz} is nonpositive, $V(k)\le V(k-1)$. Moreover, from \eqref{eq:zxk} and \eqref{eq:VV=sumzzQzz}, $V(k)=V(k-1)$ if and only if $z_i(k-1)$ $\forall i\in\mathcal{I}(k)\cup\mathcal{L}(k)$ are equal.

\section*{Proof of Proposition~\ref{pro:netwconn}}

First, suppose the graph $(\mathcal{F},\mathcal{E}_\infty)$ is connected. Pick any $k\in\mathbb{N}$ and let $k'=\inf\{\tilde{k}\ge k+1:\forall\{i,j\}\in\mathcal{E}_\infty,\;\exists\bar{k}\in[k+1,\tilde{k}]\;\text{such that}\;\{i,j\}\subset\mathcal{I}(\bar{k})\}$. Then, from \eqref{eq:Einfty}, $k'<\infty$. Due to \eqref{eq:Cinitial}, \eqref{eq:C}, and $(\mathcal{F},\mathcal{E}_\infty)$ being connected, we have $C_i(k,k')=\mathcal{F}$ $\forall i\in\mathcal{F}$. It follows from \eqref{eq:Dk} and \eqref{eq:hk} that $k'\in D_k$ and, thus, $h(k)\le k'-k<\infty$. From Definition~\ref{def:netwconn}, the network is connected under $\mathcal{A}$. Conversely, suppose the network is connected under $\mathcal{A}$, i.e., $h(k)<\infty$ $\forall k\in\mathbb{N}$. For each $k\in\mathbb{N}$, let $\tilde{\mathcal{E}}(k)=\cup_{k'=k+1}^{k+h(k)}\{\{i,j\}\subset\mathcal{F}:\{i,j\}\subset\mathcal{I}(k')\}$. Then, due to \eqref{eq:Cinitial}, \eqref{eq:C}, \eqref{eq:Dk}, and \eqref{eq:hk}, the graph $(\mathcal{F},\tilde{\mathcal{E}}(k))$ is connected. Let $\boldsymbol{\mathcal{E}}$ denote the collection of all nonempty edge sets associated with vertex set $\mathcal{F}$. Clearly, $\boldsymbol{\mathcal{E}}$ contains $2^{|\mathcal{F}|(|\mathcal{F}|-1)/2}-1$ sets and $\tilde{\mathcal{E}}(k)\in\boldsymbol{\mathcal{E}}$ $\forall k\in\mathbb{N}$. Then, $\exists\bar{\mathcal{E}}\in\boldsymbol{\mathcal{E}}$ such that $\tilde{\mathcal{E}}(k)=\bar{\mathcal{E}}$ for infinitely many $k\in\mathbb{N}$. From \eqref{eq:Einfty}, we see that $\bar{\mathcal{E}}\subset\mathcal{E}_\infty$. Therefore, the graph $(\mathcal{F},\mathcal{E}_\infty)$ is connected.

\section*{Proof of Theorem~\ref{thm:SEboun}}

Let $\mathcal{A}$ be given. From \eqref{eq:sumQ=sumQ0} and \eqref{eq:Qinitial}, $Q_i(k)\le\sum_{i\in\mathcal{M}(0)}P_i$ $\forall k\in\mathbb{N}$ $\forall i\in\mathcal{M}(k)$. Thus, \eqref{eq:Qbound} holds, i.e., each $Q_i(k)$ is bounded. To derive \eqref{eq:zbound}, suppose $\{Q_i(k)\}_{k\in\mathbb{N},i\in\mathcal{M}(k)}$ is uniformly positive definite under $\mathcal{A}$ and let $\alpha>0$ be such that $Q_i(k)-\alpha I\in\mathbb{S}_+^n$ $\forall k\in\mathbb{N}$ $\forall i\in\mathcal{M}(k)$. Then, from \eqref{eq:V} and Lemma~\ref{lem:optnonincr}, $\alpha\sum_{i\in\mathcal{M}(k)}\|z_i(k)-z\|^2\le V(k)\le V(0)$ $\forall k\in\mathbb{N}$. Therefore, \eqref{eq:zbound} is satisfied, i.e., each $z_i(k)$ is bounded.

\section*{Proof of Corollary~\ref{cor:SEboun}}

Suppose the membership dynamics \eqref{eq:M=McupJ-L} are ultimately static under $\mathcal{A}$. Then, by Definition~\ref{def:ultistat}, $\exists k\in\mathbb{N}$ such that $\forall\ell>k$, $\mathcal{M}(\ell)=\mathcal{M}(k)$. Due to \eqref{eq:Qinitial}, \eqref{eq:Qunchange}, and \eqref{eq:Qchange}, $\exists\alpha>0$ such that $Q_i(\ell)>\alpha I$ $\forall\ell\le k$ $\forall i\in\mathcal{M}(\ell)$. Due again to \eqref{eq:Qunchange}, $Q_i(\ell)=Q_i(k)$ $\forall\ell\ge k+1$ $\forall i\in\mathcal{M}(\ell)$. Hence, $\{Q_i(k)\}_{k\in\mathbb{N},i\in\mathcal{M}(k)}$ is uniformly positive definite under $\mathcal{A}$. It follows from Theorem~\ref{thm:SEboun} that \eqref{eq:Qbound} and \eqref{eq:zbound} hold.

\section*{Proof of Lemma~\ref{lem:SEVexpdecr}}

Let $\mathcal{A}$ be given. Suppose the agent network is connected under $\mathcal{A}$ at some time $k\in\mathbb{N}$, i.e., $h(k)<\infty$. If $h(k)=0$, then from \eqref{eq:hk}, \eqref{eq:Dk}, and \eqref{eq:Cinitial}, $\mathcal{M}(k)=\{i\}$ for some $i\in\{1,2,\ldots,M\}$. Also, from \eqref{eq:sumQz=sumQz0}, \eqref{eq:sumQ=sumQ0}, \eqref{eq:zinitial}, \eqref{eq:Qinitial}, and \eqref{eq:sumPz=sumq}, $z_i(k)=z$. It follows that $V(k)=0$ and, thus, \eqref{eq:V<=fV} holds. Now suppose $h(k)\in\mathbb{P}$ and consider the following:

\begin{lemma}\label{lem:zXtoeta<=zitoeta}
For any $\ell\in\mathbb{N}$, any nonempty $X\subset\mathcal{M}(\ell)$, and any $\eta\in\mathbb{R}^n$, $\sum_{i\in X}\|z_X^\ell-\eta\|_{Q_i(\ell)}^2\le\sum_{i\in X}\|z_i(\ell)-\eta\|_{Q_i(\ell)}^2$.
\end{lemma}

\begin{proof}
Due to \eqref{eq:zxk}, $\sum_{i\in X}Q_i(\ell)z_X^\ell=\sum_{i\in X}Q_i(\ell)z_i(\ell)$. Therefore, $\sum_{i\in X}(z_X^\ell)^TQ_i(\ell)\eta=\sum_{i\in X}z_i(\ell)^TQ_i(\ell)\eta$ and $\sum_{i\in X}(z_X^\ell)^TQ_i(\ell)z_X^\ell=\sum_{i\in X}z_i(\ell)^TQ_i(\ell)z_X^\ell$. Because of these two properties, $\sum_{i\in X}\|z_X^\ell-\eta\|_{Q_i(\ell)}^2-\sum_{i\in X}\|z_i(\ell)-\eta\|_{Q_i(\ell)}^2=-\sum_{i\in X}\|z_i(\ell)-z_X^\ell\|_{Q_i(\ell)}^2\le0$.
\end{proof}

\begin{lemma}\label{lem:zitozX<=zitoeta}
For any $\ell\in\mathbb{N}$, any nonempty $X\subset\mathcal{M}(\ell)$, and any $\eta\in\mathbb{R}^n$, $\sum_{i\in X}\|z_i(\ell)-z_X^\ell\|_{Q_i(\ell)}^2\le\sum_{i\in X}\|z_i(\ell)-\eta\|_{Q_i(\ell)}^2$.
\end{lemma}

\begin{proof}
Using the two properties in the proof of Lemma~\ref{lem:zXtoeta<=zitoeta}, we have $\sum_{i\in X}\|z_i(\ell)-z_X^\ell\|_{Q_i(\ell)}^2-\sum_{i\in X}\|z_i(\ell)-\eta\|_{Q_i(\ell)}^2=-\sum_{i\in X}\|z_X^\ell-\eta\|_{Q_i(\ell)}^2\le0$.
\end{proof}

Let $\alpha>0$ be such that $Q_i(\ell)-\alpha I\in\mathbb{S}_+^n$ $\forall\ell\in[k,k+h(k)]$ $\forall i\in\mathcal{M}(\ell)$. This and \eqref{eq:Qbound} imply that
\begin{align}
\alpha I<Q_i(\ell)\le\beta I,\quad\forall\ell\in[k,k+h(k)],\;\forall i\in\mathcal{M}(\ell).\label{eq:alphaI<Q<=betaI}
\end{align}
Assume, to the contrary, that \eqref{eq:V<=fV} does not hold, i.e., $V(k+h(k))>\frac{\gamma(k)}{\gamma(k)+1}V(k)$, which, due to Lemma~\ref{lem:optnonincr}, implies that $V(k)>0$. For convenience, let
\begin{align}
\epsilon=\frac{V(k)}{\gamma(k)+1}>0.\label{eq:eps}
\end{align}
Then, $V(k)-V(k+h(k))\le\epsilon$. It follows from Lemma~\ref{lem:optnonincr} that
\begin{align}
V(\ell-1)-V(\ell)\le\epsilon,\quad\forall\ell\in[k+1,k+h(k)].\label{eq:V-V<eps}
\end{align}
Due to \eqref{eq:V-V<eps}, \eqref{eq:VV=sumzzQzz}, and \eqref{eq:alphaI<Q<=betaI},
\begin{align}
&\|z_i(\ell-1)-z_{\mathcal{I}(\ell)\cup\mathcal{L}(\ell)}^{\ell-1}\|^2\le\frac{\epsilon}{\alpha},\nonumber\displaybreak[0]\\
&\quad\forall\ell\in[k+1,k+h(k)],\;\forall i\in\mathcal{I}(\ell)\cup\mathcal{L}(\ell).\label{eq:|z-z|<=eps/alpha}
\end{align}
Next, let $d_i(\ell)=\sum_{j\in C_i(k,\ell)}\|z_j(\ell)-z_{C_i(k,\ell)}^\ell\|_{Q_j(\ell)}^2$ $\forall\ell\ge k$ $\forall i\in\mathcal{M}(\ell)$. In addition, let $m(\ell)$ be the number of distinct sets in the collection $\{C_i(k,\ell)\}_{i\in\mathcal{M}(\ell)}$ $\forall\ell\ge k$. Notice from \eqref{eq:Cinitial} and \eqref{eq:C} that $1\le m(\ell)\le|\mathcal{M}(\ell)|\le M$ $\forall\ell\ge k$ and $m(\ell)\le m(\ell-1)$ $\forall\ell\ge k+1$. Moreover, let $B(\ell)=\{k\}\cup\{k'\in[k+1,\ell]:m(k')<m(k'-1)\}$ $\forall\ell\ge k+1$. Consider the following lemma:

\begin{lemma}\label{lem:dboun}
For each $\ell\in[k,k+h(k)]$,
\begin{align}
d_i(\ell)&\le(\frac{4\beta}{\alpha})^{|B(\ell)|-1}(M+1-m(\ell))\nonumber\displaybreak[0]\\
&\quad\times\Bigl(\prod_{k'\in B(\ell)}(M+1-m(k'))\Bigr)\epsilon,\quad\forall i\in\mathcal{M}(\ell).\label{eq:d}
\end{align}
\end{lemma}

\begin{proof}
By induction over $\ell\in[k,k+h(k)]$. Let $\ell=k$. For any $i\in\mathcal{M}(\ell)$, from \eqref{eq:Cinitial}, $C_i(k,\ell)=\{i\}$, which, together with \eqref{eq:zxk}, implies that $z_i(\ell)=z_{C_i(k,\ell)}^\ell$. Hence, $d_i(\ell)=0$ $\forall i\in\mathcal{M}(\ell)$. Since the right-hand side of \eqref{eq:d} is positive, \eqref{eq:d} holds for $\ell=k$. Next, let $\ell\in[k+1,k+h(k)]$ and suppose
\begin{align}
d_i(\ell\!-\!1)&\le(\frac{4\beta}{\alpha})^{|B(\ell-1)|-1}(M+1-m(\ell-1))\nonumber\displaybreak[0]\\
&\quad\times\Bigl(\prod_{k'\in B(\ell-1)\!\!\!\!\!\!\!\!\!\!\!\!\!\!\!\!\!\!\!}(M\!+\!1\!-\!m(k'))\Bigr)\epsilon,\quad\forall i\!\in\!\mathcal{M}(\ell\!-\!1).\label{eq:dl-1}
\end{align}
Below, we show that \eqref{eq:dl-1} implies \eqref{eq:d}. To do so, consider the following two mutually exclusive and exhaustive cases:

{\em Case~(I)}: $\mathcal{I}(\ell)\cup\mathcal{L}(\ell)\subset C_{i^*}(k,\ell-1)$ for some $i^*\in\mathcal{M}(\ell-1)$. Due to \eqref{eq:C}, we have $m(\ell)=m(\ell-1)$, so that $B(\ell)=B(\ell-1)$. Let $i\in\mathcal{M}(\ell)$. Suppose $i\in\mathcal{M}(\ell)-(C_{i^*}(k,\ell-1)\cup\mathcal{J}(\ell))$. Then, due to \eqref{eq:C}, \eqref{eq:z=z}, and \eqref{eq:Q=Q}, $C_i(k,\ell)=C_i(k,\ell-1)$, $z_j(\ell)=z_j(\ell-1)$ $\forall j\in C_i(k,\ell)$, and $Q_j(\ell)=Q_j(\ell-1)$ $\forall j\in C_i(k,\ell)$, implying that $d_i(\ell)=d_i(\ell-1)$. Now suppose $i\in(C_{i^*}(k,\ell-1)\cup\mathcal{J}(\ell))-\mathcal{L}(\ell)$. From \eqref{eq:C}, $C_i(k,\ell)=(C_{i^*}(k,\ell-1)\cup\mathcal{J}(\ell))-\mathcal{L}(\ell)$. Thus, from \eqref{eq:z=z}, \eqref{eq:Q=Q}, \eqref{eq:sumQz=sumQz}, and \eqref{eq:sumQ=sumQ}, we have $\sum_{j\in C_{i^*}(k,\ell-1)}Q_j(\ell-1)z_j(\ell-1)=\sum_{j\in C_i(k,\ell)}Q_j(\ell)z_j(\ell)$ and $\sum_{j\in C_{i^*}(k,\ell-1)}Q_j(\ell-1)=\sum_{j\in C_i(k,\ell)}Q_j(\ell)$. These and \eqref{eq:zxk} indicate that $z_{C_i(k,\ell)}^\ell=z_{C_{i^*}(k,\ell-1)}^{\ell-1}$. It follows from \eqref{eq:zi=zIL}, \eqref{eq:sumQ=sumQ}, \eqref{eq:z=z}, \eqref{eq:Q=Q}, and Lemma~\ref{lem:zXtoeta<=zitoeta} that
\begin{align*}
d_i(\ell)&=\sum_{j\in\mathcal{J}(\ell)\cup\mathcal{I}(\ell)}\|z_j(\ell)-z_{C_{i^*}(k,\ell-1)}^{\ell-1}\|_{Q_j(\ell)}^2\displaybreak[0]\\
&\quad+\sum_{\substack{j\in C_{i^*}(k,\ell-1)\\-(\mathcal{I}(\ell)\cup\mathcal{L}(\ell))}}\|z_j(\ell)-z_{C_{i^*}(k,\ell-1)}^{\ell-1}\|_{Q_j(\ell)}^2\displaybreak[0]\\
&=\sum_{j\in\mathcal{I}(\ell)\cup\mathcal{L}(\ell)}\|z_{\mathcal{I}(\ell)\cup\mathcal{L}(\ell)}^{\ell-1}-z_{C_{i^*}(k,\ell-1)}^{\ell-1}\|_{Q_j(\ell-1)}^2\displaybreak[0]\\
&\quad+\sum_{\substack{j\in C_{i^*}(k,\ell-1)\\-(\mathcal{I}(\ell)\cup\mathcal{L}(\ell))}}\|z_j(\ell-1)-z_{C_{i^*}(k,\ell-1)}^{\ell-1}\|_{Q_j(\ell-1)}^2\displaybreak[0]\\
&\le\sum_{j\in\mathcal{I}(\ell)\cup\mathcal{L}(\ell)}\|z_j(\ell-1)-z_{C_{i^*}(k,\ell-1)}^{\ell-1}\|_{Q_j(\ell-1)}^2\displaybreak[0]\\
&\quad+\sum_{\substack{j\in C_{i^*}(k,\ell-1)\\-(\mathcal{I}(\ell)\cup\mathcal{L}(\ell))}}\|z_j(\ell-1)-z_{C_{i^*}(k,\ell-1)}^{\ell-1}\|_{Q_j(\ell-1)}^2\displaybreak[0]\\
&=d_{i^*}(\ell-1).
\end{align*}
It follows from \eqref{eq:dl-1} that \eqref{eq:d} holds for Case~(I).

{\em Case~(II)}: $\mathcal{I}(\ell)\cup\mathcal{L}(\ell)\not\subset C_i(k,\ell-1)$ $\forall i\in\mathcal{M}(\ell-1)$. Due to \eqref{eq:C}, we have $m(\ell-1)-m(\ell)\ge1$ and $B(\ell)=B(\ell-1)\cup\{\ell\}$. Let $i\in\mathcal{M}(\ell)$. Suppose $i\in\mathcal{M}(\ell)-\Bigl(\cup_{j\in\mathcal{I}(\ell)\cup\mathcal{L}(\ell)}C_j(k,\ell-1)\cup\mathcal{J}(\ell)\Bigr)$. Then, observe from \eqref{eq:C}, \eqref{eq:z=z}, and \eqref{eq:Q=Q} that $C_i(k,\ell)=C_i(k,\ell-1)$, $z_j(\ell)=z_j(\ell-1)$ $\forall j\in C_i(k,\ell)$, and $Q_j(\ell)=Q_j(\ell-1)$ $\forall j\in C_i(k,\ell)$. Hence, $d_i(\ell)=d_i(\ell-1)$. Because of this and \eqref{eq:dl-1}, and because $\frac{4\beta}{\alpha}>1$, we have $d_i(\ell)\le(\frac{4\beta}{\alpha})^{|B(\ell)|-1}(M+1-m(\ell))\Bigl(\prod_{k'\in B(\ell)}(M+1-m(k'))\Bigr)\epsilon$. Now suppose $i\in\Bigl(\cup_{j\in\mathcal{I}(\ell)\cup\mathcal{L}(\ell)}C_j(k,\ell-1)\cup\mathcal{J}(\ell)\Bigr)-\mathcal{L}(\ell)$. Also, write $\{C_j(k,\ell-1)\}_{j\in\mathcal{I}(\ell)\cup\mathcal{L}(\ell)}$ as $\{C_{j_1}(k,\ell-1),C_{j_2}(k,\ell-1),\ldots,C_{j_p}(k,\ell-1)\}$, where $2\le p\le m(\ell-1)$. Then, from \eqref{eq:C},
\begin{align}
C_i(k,\ell)=\Bigl(\cup_{q=1}^p C_{j_q}(k,\ell-1)\cup\mathcal{J}(\ell)\Bigr)-\mathcal{L}(\ell).\label{eq:C=cupCcupJ-D}
\end{align}
Let $s_q\in C_{j_q}(k,\ell-1)\cap(\mathcal{I}(\ell)\cup\mathcal{L}(\ell))$ $\forall q\in\{1,2,\ldots,p\}$. Then, because of Lemma~\ref{lem:zitozX<=zitoeta}, \eqref{eq:C=cupCcupJ-D}, \eqref{eq:alphaI<Q<=betaI}, \eqref{eq:z=z}, \eqref{eq:zi=zIL}, the triangle inequality, \eqref{eq:|z-z|<=eps/alpha}, and \eqref{eq:dl-1}, we have
\begin{align*}
d_i(\ell)&\le\beta\sum_{\substack{j\in(\cup_{q=1}^pC_{j_q}(k,\ell-1)\\ \cup\mathcal{J}(\ell))-\mathcal{L}(\ell)}}\|z_j(\ell)-z_{\mathcal{I}(\ell)\cup\mathcal{L}(\ell)}^{\ell-1}\|^2\displaybreak[0]\\
&=\beta\sum_{q=1}^p\sum_{\substack{j\in C_{j_q}(k,\ell-1)\\ -(\mathcal{I}(\ell)\cup\mathcal{L}(\ell))}}\|z_j(\ell-1)-z_{\mathcal{I}(\ell)\cup\mathcal{L}(\ell)}^{\ell-1}\|^2\displaybreak[0]\\
&\le\beta\sum_{q=1}^p\sum_{\substack{j\in C_{j_q}(k,\ell-1)\\ -(\mathcal{I}(\ell)\cup\mathcal{L}(\ell))}}(\|z_j(\ell-1)-z_{s_q}(\ell-1)\|\displaybreak[0]\\
&\quad+\|z_{s_q}(\ell-1)-z_{\mathcal{I}(\ell)\cup\mathcal{L}(\ell)}^{\ell-1}\|)^2\displaybreak[0]\\
&\le\beta\sum_{q=1}^p\sum_{\substack{j\in C_{j_q}(k,\ell-1)\\ -(\mathcal{I}(\ell)\cup\mathcal{L}(\ell))}\!\!\!\!\!\!\!\!\!\!\!\!\!\!\!\!\!\!\!\!\!}2\Bigl((\|z_j(\ell\!-\!1)\!-\!z_{C_{j_q}(k,\ell-1)}^{\ell-1}\|\!+\!\|z_{C_{j_q}(k,\ell-1)}^{\ell-1}\displaybreak[0]\\
&\quad-z_{s_q}(\ell-1)\|)^2+\|z_{s_q}(\ell-1)-z_{\mathcal{I}(\ell)\cup\mathcal{L}(\ell)}^{\ell-1}\|^2\Bigr)\displaybreak[0]\\
&\le\beta\sum_{q=1}^p\sum_{\substack{j\in C_{j_q}(k,\ell-1)\\ -(\mathcal{I}(\ell)\cup\mathcal{L}(\ell))}}2\Bigl(2(\|z_j(\ell-1)-z_{C_{j_q}(k,\ell-1)}^{\ell-1}\|^2\displaybreak[0]\\
&\quad+\|z_{s_q}(\ell-1)-z_{C_{j_q}(k,\ell-1)}^{\ell-1}\|^2)+\frac{\epsilon}{\alpha}\Bigr)\displaybreak[0]\\
&\le\beta\sum_{q=1}^p\sum_{\substack{j\in C_{j_q}(k,\ell-1)\\ -(\mathcal{I}(\ell)\cup\mathcal{L}(\ell))}}2\Bigl(\frac{2}{\alpha}d_{j_q}(\ell-1)+\frac{\epsilon}{\alpha}\Bigr)\displaybreak[0]\\
&\le|C_i(k,\ell)|\Bigl((\frac{4\beta}{\alpha})^{|B(\ell-1)|}(M+1-m(\ell-1))\displaybreak[0]\\
&\quad\times\Bigl(\prod_{k'\in B(\ell-1)}(M+1-m(k'))\Bigr)\epsilon+\frac{2\beta}{\alpha}\epsilon\Bigr)\displaybreak[0]\\
&\le|C_i(k,\ell)|(\frac{4\beta}{\alpha})^{|B(\ell)|-1}(M+1-m(\ell))\displaybreak[0]\\
&\quad\times\Bigl(\prod_{k'\in B(\ell-1)}(M+1-m(k'))\Bigr)\epsilon\displaybreak[0]\\
&=|C_i(k,\ell)|(\frac{4\beta}{\alpha})^{|B(\ell)|-1}\Bigl(\prod_{k'\in B(\ell)}(M+1-m(k'))\Bigr)\epsilon.
\end{align*}
This, along with the fact that $|C_i(k,\ell)|\le M+1-m(\ell)$, implies that $d_i(\ell)\le(\frac{4\beta}{\alpha})^{|B(\ell)|-1}(M+1-m(\ell))\Bigl(\prod_{k'\in B(\ell)}(M+1-m(k'))\Bigr)\epsilon$. Therefore, \eqref{eq:d} holds for Case~(II).
\end{proof}

Since $C_i(k,k+h(k))=\mathcal{M}(k+h(k))$ $\forall i\in\mathcal{M}(k+h(k))$, we have $m(k+h(k))=1$. Moreover, $|B(k+h(k))|\le\min\{h(k)+1,M\}$ and $\Pi_{k'\in B(k+h(k))}(M+1-m(k'))\le\min\{M^{h(k)+1},M!\}$. Furthermore, note from \eqref{eq:sumPz=sumq}, \eqref{eq:sumQz=sumQz0}, \eqref{eq:sumQ=sumQ0}, \eqref{eq:zinitial}, and \eqref{eq:Qinitial} that $z=z_{\mathcal{M}(\ell)}^\ell$ $\forall\ell\in\mathbb{N}$, implying that $d_i(k+h(k))=V(k+h(k))$ $\forall i\in\mathcal{M}(k+h(k))$. It follows from Lemma~\ref{lem:dboun} and \eqref{eq:eps} that $V(k+h(k))\le\gamma(k)\epsilon\le\frac{\gamma(k)}{\gamma(k)+1}V(k)$, which contradicts the assumption that \eqref{eq:V<=fV} is violated. Consequently, \eqref{eq:V<=fV} holds.

\section*{Proof of Theorem~\ref{thm:SEasymconv}}

Let $\mathcal{A}$ be given. Suppose the agent network is connected under $\mathcal{A}$, i.e., $h(k)<\infty$ $\forall k\in\mathbb{N}$, and $\{Q_i(k)\}_{k\in\mathbb{N},i\in\mathcal{M}(k)}$ is uniformly positive definite under $\mathcal{A}$. Let $\alpha>0$ be such that $Q_i(k)-\alpha I\in\mathbb{S}_+^n$ $\forall k\in\mathbb{N}$ $\forall i\in\mathcal{M}(k)$. Then, \eqref{eq:limz=z} holds if and only if $\lim_{k\rightarrow\infty}V(k)=0$. To show that $\lim_{k\rightarrow\infty}V(k)=0$, note from \eqref{eq:V} and Lemma \ref{lem:optnonincr} that $(V(k))_{k=0}^{\infty}$ is nonnegative and non-increasing. Thus, $\exists c\ge0$ such that $\lim_{k\rightarrow\infty}V(k)=c$. To show that $c=0$, assume, to the contrary, that $c>0$. Let $\epsilon=\frac{c}{\gamma(k)}$, where $\gamma(k)$ is defined in Lemma~\ref{lem:SEVexpdecr}. Then, $\exists k\in\mathbb{N}$ such that $c\le V(\ell)<c+\epsilon$ $\forall\ell\ge k$. However, by Lemma~\ref{lem:SEVexpdecr}, we have $V(k+h(k))<\frac{\gamma(k)}{\gamma(k)+1}(c+\epsilon)=c$, which contradicts the inequality $c\le V(\ell)$ for $\ell=k+h(k)$. Therefore, $c=0$, i.e., $\lim_{k\rightarrow\infty}V(k)=0$, so that \eqref{eq:limz=z} holds.

\section*{Proof of Theorem~\ref{thm:SEconvrate}}

Let $\mathcal{A}$ be given. Suppose the agent network is uniformly connected under $\mathcal{A}$, i.e., $h^*<\infty$, and $\{Q_i(k)\}_{k\in\mathbb{N},i\in\mathcal{M}(k)}$ is uniformly positive definite under $\mathcal{A}$. Let $\alpha>0$ be such that $Q_i(k)-\alpha I\in\mathbb{S}_+^n$ $\forall k\in\mathbb{N}$ $\forall i\in\mathcal{M}(k)$. Note that $\gamma^*\ge\gamma(k)$ $\forall k\in\mathbb{N}$. Then, it follows from \eqref{eq:hstar}, Lemma~\ref{lem:optnonincr}, and Lemma~\ref{lem:SEVexpdecr} that $\forall\ell\in\mathbb{N}$, $V((\ell+1)h^*)\le V(\ell h^*+h(\ell h^*))\le\frac{\gamma^*}{\gamma^*+1}V(\ell h^*)$, which implies that $V(\ell h^*)\le\Bigl(\frac{\gamma^*}{\gamma^*+1}\Bigr)^\ell V(0)$. Due again to Lemma~\ref{lem:optnonincr}, \eqref{eq:V<=V0f} holds. In addition, from \eqref{eq:V}, $\alpha\|z_i(k)-z\|^2\le V(k)$ $\forall k\in\mathbb{N}$ $\forall i\in\mathcal{M}(k)$. Therefore, \eqref{eq:|z-z|<=V0/alphaf} is satisfied.

\bibliographystyle{IEEEtran}
\bibliography{paper_r1}

\end{document}